\documentclass[11pt, a4paper]{article}

\pdfoutput=1
\RequirePackage[numbers]{natbib}

\usepackage{hyperref}
\hypersetup{bookmarks=false}
\usepackage{amsthm,amsfonts,amssymb,amsmath,natbib}
\usepackage{mathrsfs}

\usepackage{tikz}
\usepackage{tikz-3dplot}

\usepackage{multirow}

\usepackage{listings}
\usepackage[latin1]{inputenc}
\usepackage[T1]{fontenc}
\usepackage[english]{babel}

\usepackage{subcaption}

\usepackage{articlestyle}

\usepackage{authblk}

\usepackage{fullpage}

\usepackage{ifpdf}
\usepackage{units}

\usepackage{mathtools}
\mathtoolsset{showonlyrefs}

\usepackage{xr}

\usepackage[normalem]{ulem}

\title{Deformed SPDE models with an application to spatial modeling of significant wave height}
\author{Anders Hildeman}
\author{David Bolin}
\author{Igor Rychlik}
\date{}

\affil{Department of Space-, Earth-, and Environmental Sciences, Chalmers University of Technology, Sweden}
\affil{Computer, Electrical and Mathematical Science and Engineering Division, King Abdullah University of Science and Technology, Saudi Arabia}
\affil{Department of Mathematical Sciences, Chalmers University of Technology and University of Gothenburg, Sweden}

\begin{document}

\maketitle

\setcounter{tocdepth}{2}

\begin{abstract}
\label{sec:abstract}
A non-stationary Gaussian random field model is developed based on a combination of the stochastic partial differential equation (SPDE) approach and the classical deformation method. With the deformation method, a stationary field is defined on a domain which is deformed so that the field becomes non-stationary. We show that if the stationary field is a Mat\'ern field defined as a solution to a fractional SPDE, the resulting non-stationary model can be represented as the solution to another fractional SPDE on the deformed domain. By defining the model in this way, the computational advantages of the SPDE approach can be combined with the deformation method's more intuitive parameterisation of non-stationarity. In particular it allows for independent control over the non-stationary practical correlation range and the variance, which has not been possible with previously proposed non-stationary SPDE models.

The model is tested on spatial data of significant wave height, a characteristic of ocean surface conditions which is important when estimating the wear and risks associated with a planned journey of a ship. The model parameters are estimated to data from the north Atlantic using a maximum likelihood approach. The fitted model is used to compute wave height exceedance probabilities and the distribution of accumulated fatigue damage for ships traveling a popular shipping route. The model results agree well with the data, indicating that the model could be used for route optimization in naval logistics.

\end{abstract}

\section{Introduction}
\label{sec:introduction}
The probability distribution of ocean waves at a location in space and time is often referred to as the \textit{sea state} of that location. 
The ability to model the sea state provides important and sometimes necessary information for risk assessment and prediction in naval logistics and marine operations.
One of the most important parameters for characterizing the sea state is the \textit{significant wave height}. The significant wave height was traditionally defined as the average wave height among the one third highest individual waves \citep{lit:munk}. This definition was intended to mathematically express the average wave-height estimated by a ``trained observer'' such as a sailor or fisherman. A more common definition today, that gives values of similar magnitude but is easier to work with mathematically, is that the significant wave height is four times the standard deviation of the sea elevation for a fixed point in space and time. 
Even though the significant wave height only provides limited information about the sea state it is sufficient for many situations and applications.
The safety of a naval structure may depend on extreme and rare events such as loads which exceed the strength of components, or on everyday load variability that may cause changes in the properties of material, e.g., cracking (fatigue) or other types of aging processes. Typically such loads can be approximated using significant wave height~\citep{lit:mao, lit:lewis}. 

The sea state is rarely constant over large spatial regions or for any longer time periods. 
Instead the significant wave height can be regarded as the parameter that describes the evolution in time and space of the local wave energy~\citep{lit:baxevani3}.
This means that the significant wave height is a varying function in space and time. As such, and with all uncertainties and lack of information associated with predicting future sea states, the significant wave height over a larger region of space and/or time has to be modeled as a random process (in space and/or time). 

For small regions of the north Atlantic, the sea state during a time interval of about an hour can be considered as constant and distributed according to a log-Gaussian distribution~\citep{lit:jasper, lit:ochi}.
Also for other ocean regions similar transformations to Gaussianity are often possible.
For applications only concerned with the sea state at one point in space during a shorter interval of time, this is sufficient. However, many applications are concerned with quantities acquired over larger regions or paths on the spatial domain and possibly during a longer time duration. Examples of such are spatial interpolation/data assimilation of sea states from measurements, or risks and wear analysis of ships given a planned route.
When more than one point should be considered simultaneously it is often necessary to account for the dependency among these points, i.e., a spatial/temporal model is required. 
Moreover, for larger spatial regions this model has to be non-stationary~\citep{lit:baxevani3, lit:ailliot}.

\citet{lit:ailliot} considered non-stationary marginal standardizations and then assumed a stationary Gaussian random field model in space. 
\citet{lit:baxevani3, lit:baxevani4} developed a non-stationary spatial model based on the logarithm of significant wave height being a Gaussian random field. The covariance model was defined as a non-stationary process-convolution~\citep{lit:higdon} using a Gaussian kernel. Hence, the non-stationarity was characterized by one spatially-varying parameter controlling the correlation length between point pairs.
The choice of a Gaussian kernel is limiting the types of possible spatial correlations that can be modeled. Moreover, it yields infinitely mean-square differentiable realizations of the significant wave height, which is theoretically unrealistic since it implies that observing a small region of the ocean thoroughly would give perfect information about the values over the whole ocean~\citep{lit:stein}.
This issue was noted by~\citet{lit:paciorek} who instead introduced a non-stationary Mat\'{e}rn model based on process-convolution. This model can be employed to allow for control of the smoothness, i.e., differentiability.
However, the non-stationary process-convolution model is not theoretically suitable for modeling the significant wave height, as well as many other spatial variables. The reason being that the covariance between a pair of points are only dependent on the local spectral densities at those two specific points. In most variables studied in spatial statistics, shared information between a pair of points has to flow through the spatial domain, i.e., a spatially continuous Markov property. This is clearly true for significant wave height as well, where wave energy travels in space and time and is generated by wind friction over large regions of water. Strong correlation between two points in space indicates that they are both affected by the same wind/wave systems. This suggests that correlation between a pair of points should be dependent, not only on the parameter values at the two points themselves, but on the parameter values at all points corresponding to possible paths between the two points. 

Besides non-stationary process-convolution there are other methods for acquiring \ non-stationarity. \citet{lit:sampson2} divides them into three categories, viz., smoothing-based methods, basis function methods, and deformation methods.  
Smoothing methods acquires non-stationarity by introducing a convolution over a set of independent random fields. Compared to the process-convolution methods the kernel itself is stationary but the parameters of the independent random fields are indexed by points in space, yielding a non-stationary smoothing over the set of independent random fields. 
The basis function methods decomposes the random field into a, possibly infinite, sum of spatially varying functions with associated random coefficients. The spatial modeling is then transformed into modeling of the dependencies between the set of random coefficients. 
The deformation methods~\citep{lit:sampson} allow for non-stationarity by assuming that the spatial domain can be deformed using a bijective mapping into a \textit{deformed space}. If the random field is mapped to this \textit{deformed space} it will be stationary. The problem of non-stationarity is hence divided into that of finding a suiting bijective mapping and a suiting stationary model. 
The deformation method has the attractive continuous Markov property that the process-convolution methods were lacking (as long as the stationary covariance function has this property).
The problem with the deformation method is parameterizing this deformation in an adequate way that is flexible enough while avoiding \textit{folding}, see \citet{lit:sampson}.

A promising new method for modeling non-stationarity is the non-stationary stochastic partial differential equation (SPDE) models of \citet{lit:fuglstad, lit:fuglstad2, lit:bolin1}. These are closely connected to the deformation method but does not require explicit parameterization of the deformation. More importantly, compared to the non-stationary models listed above, these models can handle complicated spatial domains, such as coast lines, while having beneficial computational properties that allow for high resolution modeling of large regions of open ocean.
However, a problem with these models are that, for the non-stationary case, the spatially varying correlation parameters and marginal parameters becomes intertwined and their interpretation is hence lost. 
This means that one cannot change the marginal variance without affecting the correlation structure and vice versa, leading to some possible problems, viz.,
\begin{itemize}
    \item The clear interpretation of each parameter is lost since they jointly represent both the spatial correlation behavior and the marginal variances.
    \item All parameters have to be fitted jointly. This increases the computational complexity of parameter estimation, as compared to estimating the marginal variance first and then the spatial correlation.
    \item The identifiability of the model is lost since more than one set of parameters can correspond to exactly the same model. 
\end{itemize}
Our contribution is deriving a non-stationary SPDE-model which solves the issues of \citet{lit:fuglstad, lit:fuglstad2} by combining the deformation method \citep{lit:sampson} and the Mat\'{e}rn SPDE-approach \citep{lit:lindgren}.
The model maintains all the computational advantages of \citet{lit:fuglstad}. Furthermore, due to our derived link between the deformation method and the non-stationary SPDE, it is possible to calculate properties of the non-stationary random field that have only been possible for stationary fields before. As an example, in this work we derive an analytical formula for exceedance probabilities when traveling a specified path over the spatial domain.
We use our results to model significant wave height in the north Atlantic spatially and to analyze the risks and wear of a ship travelling a transatlantic route. 

In this paper we will denote a point in space by $\psp$, possibly by an index if several points are present simultaneously. Likewise, a point in time will be denoted by $t$ and possibly an index if necessary. The spatial domain will be denoted as $\gspace$, or sometimes when needed as $\dspace$. The vector ``dot-product'' is denoted as $\cdot$, the gradient operator of a function in several dimensions as $\nabla$, the Laplacian as $\Delta = \nabla \cdot \nabla$, the identity matrix as $\eye$, and the domain (of an operator) as $\mathscr{D}$.

The structure of the paper is as follows. In Section~\ref{sec:data} the data and the problems of the case study are presented. This section also act as a motivation to why a spatial model of significant wave height is needed. Also, an upper bound for exceedance probabilities of significant wave height when traveling a shipping route is given. 
In Section~\ref{sec:model} the proposed model is derived. A theorem proving the relationship between the model and a deformation of a Gaussian Mat\'{e}rn random field is presented. Also, a semi-explicit formula for  exceedance probabilities of a given curve using the proposed model is presented.
Section~\ref{sec:estimation} introduces how the model can practically be parameterized and estimated given data.
The results of the case study are presented in Section~\ref{sec:results}. This include a comparison between the non-stationary model and a stationary anisotropic Mat\'{e}rn model as well as the performance when analyzing accumulated fatigue damage and exceedance probabilities for a transatlantic route in the north Atlantic. 
Finally, Section~\ref{sec:discussion} concludes with a discussion of the results and potential further applications and extensions of the model. The manuscript also contains four appendices which present further details on the model and the estimation procedure as well as all proofs.

\section{Case study: spatial distribution of significant wave height}
\label{sec:data}
We want to model the significant wave height, commonly denoted as $H_s$, spatially over a large region of open ocean. 
Getting consistent high quality measurements of $H_s$ over a large spatial region, with a good spatial resolution, is not easy. In fact, this is one of the reasons for developing a spatial model of $H_s$, to interpolate measured values to the whole of a spatial domain. 
However, in this paper we are mainly concerned with the effect that encountered sea states will have on ships sailing a specified route. Although forecasting methods for wind and waves are constantly evolving, predicting the sea state more than two weeks in advance is practically impossible. At the same time, applications such as policy making, ship design, and business strategy for naval logistics require decisions to be made further into the future than two weeks. Hence, making use of the spatial probability distribution of sea states together with modern ship models can yield important information for the applications mentioned above. 

To be able to fit and compare models with data, we will use data from the ERA-Interim global atmospheric reanalysis~\citep{lit:dee}, performed by the European Centre for Medium-Range Weather Forecasts (ECMWF).
The ERA-Interim dataset provides hindcasted data of several important atmospheric and oceanic variables on a global longitude-latitude grid.
The spatial resolution of the dataset is approximately 80 km and data is available from 1979 to present. 
The variable in the dataset we will work with is \textit{significant wave height of wind and ground swells}, which is available at a temporal resolution of 6 hours. 
The values produced for a fixed point in time are based on measurements from a temporal window of $\pm$ 6 hours. Hence, there is a 3 hour overlap in between data from consecutive points in time which will cause a smoothing effect. 
We will not model the temporal evolution and therefore we want to approximate data from different points in time as independent. Therefore, we choose to thin the data to a temporal resolution of 24 hours (only keeping every fourth time stamp).  

Even though we do not believe that the ERA-Interim dataset represents the true values of $H_s$ perfectly, it is probably the most accurate representation thereof available.  
In  previous  work~\citep{lit:baxevani2} we  have  used  satellite  measurements to estimate a spatial log-normal Gaussian field with Gaussian correlations for $H_s$ variability. When ECMWF data have been available we have re-estimated that model at several locations using the ERA-Interim dataset. The two datasets have given similar results which indicates that the ERA-Interim dataset performs well for our purposes. We will interpret the values of the dataset as point observations at the grid locations.

We choose the north Atlantic as our spatial domain since this region provides some of the most important trading routes of naval logistics (both historically and at present) and is known for its unpredictable storms. Furthermore, the region is known to produce data for which the bulk of the pointwise marginal distributions for any chosen month are well approximated by log-Gaussian distributions~\citep{lit:jasper, lit:ochi}. 
We also restrict the analysis to the month of April, but using data from all of the 38 years (1979-2017). The reason for this restriction being that data from different months will be distributed differently due to the effects of the annual cycle, which we do not model in this work.
The month of April was chosen since, among all the months of the year, for this month the tails of the logarithmized marginal data seemed to fit a normal distribution the best. However, any other month could have been chosen instead since they all showed good agreement with normal distributions.
From here on, we will refer to the logarithm of $H_s$ as $X$. Our model, as well as former models~\citep{lit:baxevani2, lit:baxevani3, lit:baxevani4}, are based on modeling $X$ as a Gaussian random field.

The mean and variance of the logarithmized data can be seen in Figure~\ref{fig:meanVarFields}. Clearly the wave height is diminishing near the coasts while variance increases. Non-stationarity in the correlation structure can be observed in the left column of Figure~\ref{fig:correlation} portraying the empirical correlation function of log-$H_s$ between three different reference points and all other points on the observational domain.
Apparently, the point close to the coast of USA is showing an anisotropic pattern with the principal axis on the diagonal. Contrary to this, the spatial correlation of the mid Atlantic and at the coast of Europe have the principal axis in the east-west direction. It should be noted that the data is portrayed in the longitude-latitude coordinate system and other projections would yield different non-stationarity. 

\begin{figure}
	\centering
	\begin{subfigure}{0.48\textwidth}
		\includegraphics[width = \textwidth]{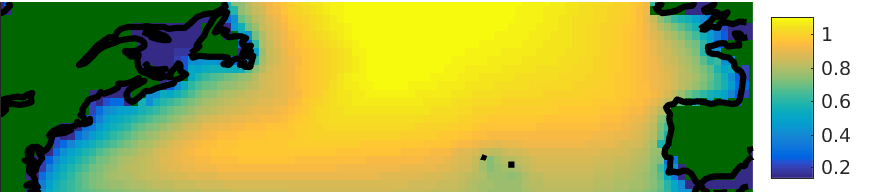}
	\end{subfigure}
	\begin{subfigure}{0.48\textwidth}
		\includegraphics[width = \textwidth]{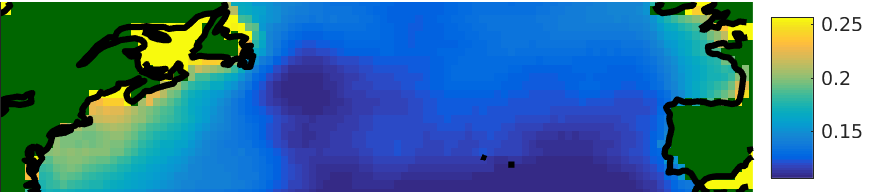}
	\end{subfigure}
	\caption{Sample mean (left) and sample variance (right) of $\log H_s$.  }
	\label{fig:meanVarFields}
\end{figure}


\subsection{Fatigue damage accumulated by a vessel}
\citet{lit:mao} presented a formula approximating the fatigue damage accumulated by a ship sailing under a stationary sea state. The formula is a function of the wave period and the significant wave height and is based on the narrow-band approximation~\citep{lit:rychlik}. 
According to this formula, the expected instantaneous fatigue damage to a ship, en route, as a function of the significant wave height, zero crossing period, ship velocity, and wave velocity is
\begin{align}
d(t) \approx
\frac{0.47 C^{\beta} H_s^\beta(\psp(t)), t}{\gamma} \left( \frac{1}{T_z(\psp(t), t)} - \frac{2\pi V(t)\cos \theta(t)}{gT_z^2(\psp(t), t)} \right),
\label{eq:fatigueDamage}
\end{align}
where $g$ is the gravitational constant ($\approx 9.81$), $T_z$ is the mean zero crossing period of the waves, $V$ is the speed of the ship, and $\theta$ is the angle between the heading of the ship and the direction of the traveling waves. Further, $\gamma$ and $\beta$ are constants dependent on the material of the ship and $C$ is a constant depending on the ship design, see \citep{lit:mao} for details. 
Note that the formula concerns a route in space-time parametrized by time, $t$, i.e. $H_s(\psp(t), t)$ is the significant wave height at point $\psp(t)$ at time $t$ and so on. 

From the formula we see that increasing wave height and diminishing wave period will increase damage. In general, the conclusion is that high waves have a greater impact on the fatigue damage but also that steep waves are more damaging than flat ones. 
Since $T_z$ is not modeled in this work, we follow the procedure of~\citet{lit:mao, lit:podgorski} and approximate it as $T_z \approx 3.75 \sqrt{H_s}$, making the formula in \eqref{eq:fatigueDamage} dependent only on the values of $H_s$, for which we are developing a spatial model.
The total damage up until time $T$ is computed by numerically integrating $d$,
\begin{align}
D(T) = \int_0^T d(t)dt \approx \sum_i d_i \Delta t,
\label{eq:fatigueDamageInt}
\end{align}
where $d_i$ denotes a sequence of evaluations of $d(t)$, $\Delta t$ apart.
It should be noted that the formula of Equation~\eqref{eq:fatigueDamage} only concerns the conditional expectation of fatigue damage given the sea state parameters. However, the variability conditioned on the sea states are very small in comparison to the variability without conditioning on the sea states.

Given a stochastic model of the significant wave height, from which realizations can be generated quickly, the distribution of accumulated fatigue damage can be approximated by Monte-Carlo simulations.

 \subsection{Extreme loads}
 When using a significant wave height model to analyze risks for a ship, it is often of interest to calculate the probabilities of exceeding a given threshold value of $H_s$ when sailing along a chosen ship route. 
For instance, extreme sea conditions characterized by very large significant wave heights can lead to breakage of the hull instantaneously or over a very short time span. Therefore, it is important to have control of the risks of encountering such extreme sea states and possibly reroute the ship if the risk turns out to be too high.

Mathematically this corresponds to the probability of the random field exceeding a threshold on a given curve in the spatial domain.
Consider a curve $\gamma(t) = [\psp(t), t]^T$, where $\psp$ a point in the spatial domain of the ocean and $t$ a point in in an interval of time. 
Further, define $X_\gamma(t) = X(\psp(t),t)$ as the random process acquired by evaluating the spatio-temporal field of significant wave height on the curve $\gamma$.
The corresponding mean function and marginal variance function of $X_{\gamma}$ are denoted as $\mu_{\gamma}(t)$ and $\std_{\gamma}(t)$.
For analyzing the risk of encountering high significant wave heights on a ships journey one needs the exceedance probability function, $\prob{\max_{t\in[0,T]} X_{\gamma}(t) > u}$.

Just as with accumulated fatigue damage, given a stochastic model of the significant wave height, it would be possible to approximate the exceedance probability using Monte-Carlo simulations.
However, when analyzing the risks to a ship, the probability of encountering really high values of significant wave height are typically of interest. High values have a low probability of occurring and hence, the sample size of the Monte-Carlo simulation would need to be large to yield a reliable estimate of the exceedance probability. A better alternative would be to use an analytical formula if available. 
Although we are not able to derive a sharp formula for the exceedance probability, we can derive an upper limit. 

\begin{prop}
\label{thm:exceedanceProb} 
Assume that $X(\psp,t)$ is a Gaussian random fields that is differentiable in mean square sense and where both its marginal parameters, $\mu(\psp,t)$ and $\std(\psp,t)$, are  differentiable with continuous partial derivatives. Further assume that the curve $\gamma$, is differentiable with continuous partial derivatives. Then,
\begin{align}
&\prob{ \max_{t\in[0,T]} X_\gamma(t)>u } \le 
\int_{0}^T \left( \std_{\dot{W}}(t) \phi\left( \frac{a(t)}{\std_{\dot{W}}(t)} \right) + a(t) \right) \phi\left(\frac{u-\mu_\gamma(t)}{\std_\gamma(t)}\right) dt\\
&\qquad + 
\Phi\left( \frac{\mu_\gamma(0)-u}{\std_\gamma(0)} \right).
\end{align}
Here, $\phi$ and $\Phi$ denotes the pdf and cdf of the standard normal distribution, 
$a(t) := \frac{u-\mu_\gamma(t)}{\std_\gamma(t)^2}\dot{\std}_\gamma(t) + \frac1{\std_\gamma(t)}\dot{\mu}_\gamma(t)$, the dot above the symbol of a function denotes the derivative w.r.t. time.
The stochastic process $W(t)$ is a marginal standardization of $X_{\gamma}(t)$, i.e., $W(t) := \frac{X_\gamma(t)-\mu_\gamma(t)}{\std_\gamma(t)}$, and $\std_{\dot{W}}(t)$ is the standard deviation of the mean square derivative of $W(t)$.
\end{prop}
The proof of the proposition can be found in Appendix \ref{sec:proofs}. 
The upper bound is a good approximation of the actual exceedance probability for large values of $u$, which is our main concern when considering risks to ships.

It is worth noticing that the only spatial dependency in this formula is in the marginal variance of the derivative process $\dot{W}(t)$, i.e., $\std_{\dot{W}}(t)$.

\section{The deformed SPDE model}
\label{sec:model}

Sections~\ref{sec:introduction} and~\ref{sec:data} have given the background and motivation to why a non-stationary Gaussian random field model with low computational cost is needed for modeling the logarithm (or for some regions of the worlds oceans possibly another marginal transformation into Gaussianity) significant wave height, $X$ over large regions of ocean. This section introduces our proposed model as well as the main theoretical results of this work. It should be noted that this model, and the results within, are in no way exclusive for modeling significant wave height. The results of this work should be of importance for any one concerned with modeling non-stationary Gaussian random fields using the SPDE-approach~\citep{lit:lindgren}. 

A Gaussian random field is completely characterized by its mean and covariance functions. 
For simplicity, and without loss of generality, we will in this section assume a centered Gaussian random field, $X$, and therefore only be concerned with the model for the covariance structure.  
The Mat\'{e}rn covariance function is a popular choice due to its flexibility using only three easily interpretable parameters. For $\psp_1,\psp_2 \in \dspace \subseteq \mathbb{R}^d$, it is defined as
\begin{align}
\Cov(\psp_1, \psp_2) &= \sigma^2 C_\nu(\kappa \|\psp_1 -  \psp_2\|), \quad C_\nu( h)  = \frac{1}{2^{\nu-1}\Gamma(\nu)}\left( h \right)^{\nu}K_{\nu}(h).
\label{eq:matern}
\end{align}
For a random field with a Mat\'ern covariance function, $\sigma^2$ is the marginal variance, $\kappa>0$ controls the practical correlation range, and $\nu>0$ determines the smoothness. \citet{lit:whittle} showed that if $\mathcal{D}:= \mathbb{R}^d$, a centered Gaussian random field $\rv(\psp)$ with a Mat\'{e}rn covariance is equal to the solution of the SPDE, 
\begin{align}
\left( \damp^2 - \Delta \right)^{\alpha/2}  (\tau \rv) = \noise, \quad  \mbox{in }\mathcal{D}, 
\label{eq:lindgrenSPDE}
\end{align}
where $\Delta$ is the Laplacian and $\noise$ is Gaussian white noise. The fractional exponent $\alpha$ is related to the smoothness parameter through $\alpha = \nu + d/2$, and 
$$
\preci = \sqrt{ \frac{\Gamma(\nu)}{\Gamma(\nu+d/2)(4\pi)^{d/2} } }\frac{1}{\kappa^{\nu} \std},
$$
is a constant needed to get the desired marginal variance, $\std^2$. 

\citet{lit:lindgren} used the SPDE formulation in combination with the finite element method to compute Gaussian Markov random field approximations of Gaussian Mat\'ern fields for bounded domains, $\dspace \subset \R^n$. The computational benefits of this approach, as well as the fact that it allows for several extensions of stationary Mat\'ern fields on $\mathbb{R}^d$, has made it highly popular.   

One important extension is to non-stationary models. \citet{lit:fuglstad, lit:fuglstad2} proposed a non-stationary extension based on replacing the differential operator $\left( \damp^2 - \Delta \right)^{\alpha/2}$ in \eqref{eq:lindgrenSPDE} by $\left( \damp^2(\psp) - \nabla \cdot H(\psp) \nabla \right)$. They assumed that $H$ was a continuously differentiable and uniformly positive definite matrix-valued function and that $\kappa$ was a continuous function. \citet{lit:bolin3} proposed a larger class of models with less restrictive assumptions on $\kappa$ and $H$, and also allowed for a fractional power of the non-stationary operator as in the original stationary case. One can define random fields with a wide variety of non-stationary behaviours using differential operators on this form. However, the drawback is that $\kappa$ and $H$ jointly control local anisotropy, correlation range, and marginal variance. That means that the interpretability of the parameters from the original SPDE-approach is lost. Furthermore, in general one cannot derive a formula for the marginal variance of the non-stationary model. 

Because of these issues, we will now derive a new SPDE-based model inspired by the deformation method of~\citet{lit:sampson}. The deformation method consider the observational domain as a subset of some space, $\gspace$, being a deformation of another space, $\dspace$. Here, by deformation we mean that there exists a bijection, $\warp : \dspace \mapsto \gspace$.
A random field that is stationary and isotropic in $\dspace$ can then be mapped to $\gspace$ by $\warp$. The function $\warp$ characterizes the non-stationarity and/or anisotropy of the random field. See Figure \ref{fig:deformAnisotropicExample} for a sketch of the procedure of mapping a stationary and isotropic random field to $\gspace$. 

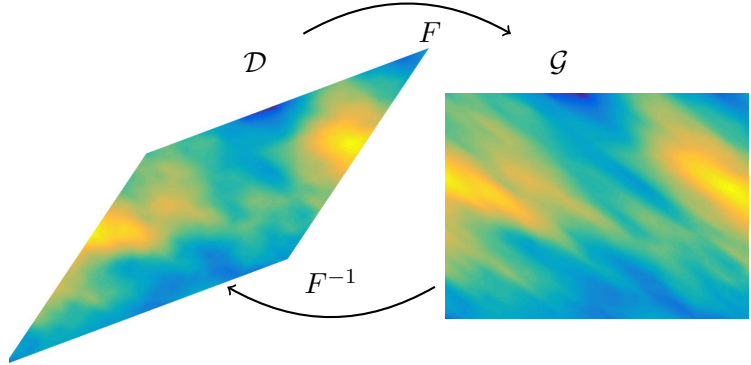
\begin{figure}[t]
\centering
\begin{tikzpicture}[scale=1]

    \node (Dimg) at (0,0) 
    {\includegraphics[width=0.35\textwidth]{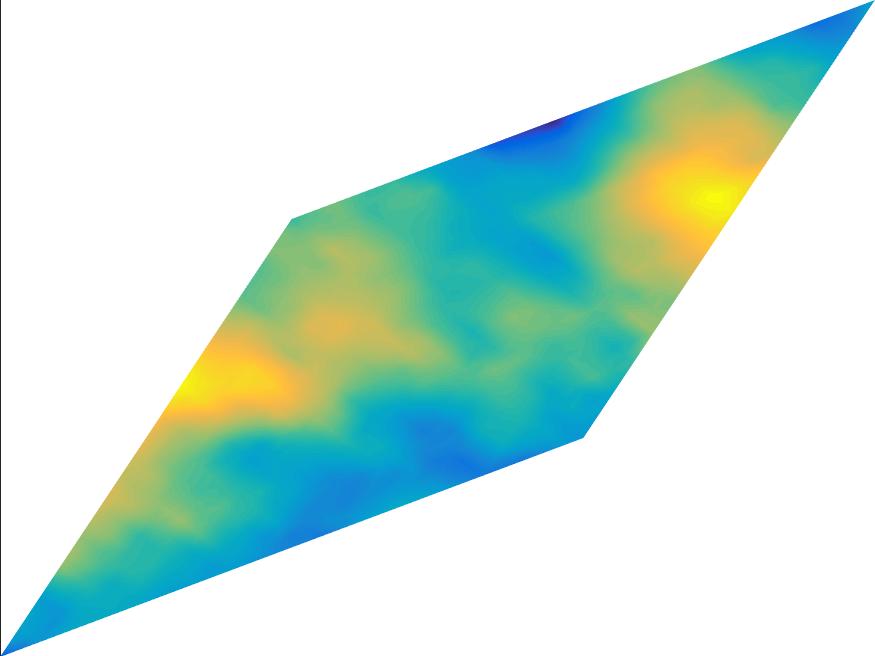}};
    \node (Gimg) at (5,0)
    {\includegraphics[width=0.25\textwidth]{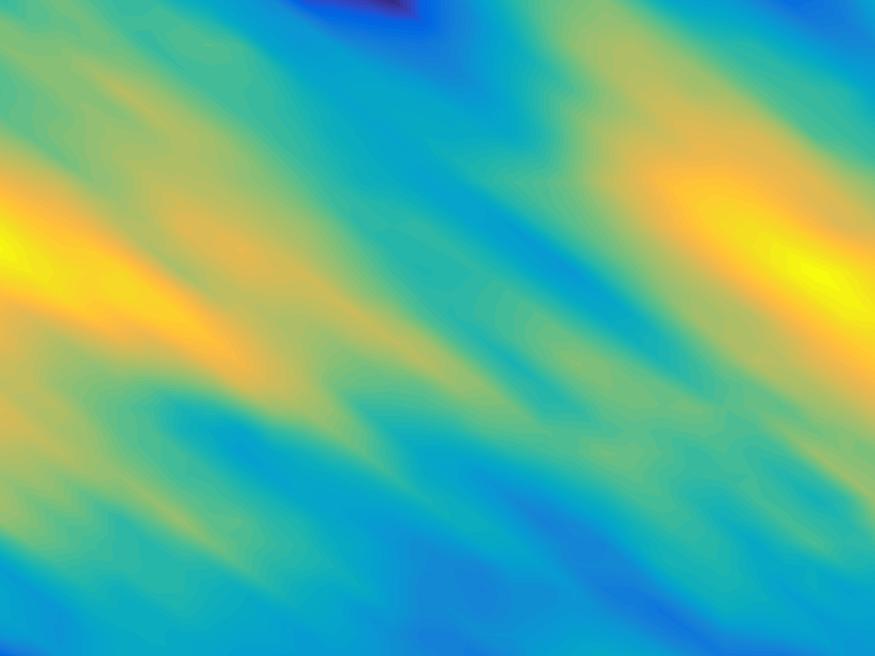}};
    
    \draw [fill=white, draw=none] (-3.2,-2.1) rectangle (-2.75,2.1);
    
    \node (D) at (0.5,1.9) {$\mathcal{D}$};
	
	\node (G) at (4.5,1.9) {$\mathcal{G}$};
    
	\node (F) at (2.8,2.3) {$F$};		
	\node (Finv) at (1.5,-1) {$F^{-1}$};		
	\node (a) at (1,2.2) {};
	\node (b) at (4,2.2) {};
	\node (c) at (3,-1) {};
	\node (d) at (0,-1) {};
	\draw[->, thick] (a) edge [bend left=30] (b);		
	\draw[->, thick] (c) edge [bend left=30] (d);

\end{tikzpicture}
\caption{Realization of anisotropic Gaussian random field using deformation method.}
\label{fig:deformAnisotropicExample}
\end{figure}

Our idea is that we have the standard Mat\'{e}rn model on $\dspace$. By the deformation this will yield a non-stationary model with the same beneficial properties as the original stationary SPDE-approach. Moreover, instead of working with a SPDE in some unknown $\dspace$-space, we want to acquire the corresponding SPDE on $\gspace$. Then, a finite element approximation can be considered directly on the observational domain of $\gspace$.
To more rigorously define the model, let $\mathcal{D}\subset \mathbb{R}^d$ be an open, bounded, convex polytope and consider the SPDE,
\begin{equation}\label{eq:Dspde}
(I - \hat{\nabla}\cdot \hat{\nabla})^{\alpha/2} \hat{X} = \hat{\noise} \quad  \mbox{in }\mathcal{D}. 
\end{equation}
The notation with a hat, such as in $\hat{X}$, denote that we are considering a formulation in spatial coordinates of $\mathcal{D}$; as opposed to the notation without a hat denoting the $\mathcal{G}$-space counterpart.
Here, $\hat{\mathcal{W}}$ is Gaussian white noise on $\mathcal{D}$ and $\hat{\mathcal{L}} = (I - \hat{\nabla}\cdot \hat{\nabla)} : \mathscr{D}(\hat{\mathcal{L})}\subset L_2(\mathcal{D}) \rightarrow L_2(\mathcal{D})$ is equipped with homogeneous Dirichlet boundary conditions on $\pd\mathcal{D}$.  By Proposition 3.1 in~\citet{lit:bolin3}, this equation has a unique solution in $L_2(\Omega; L_2(\mathcal{D}))$ whenever $\alpha>d/2$. 

We are now interested in defining a non-stationary random field, $X$, through the deformation approach. The following theorem shows that we can obtain $X$ as the solution to a SPDE directly on $\mathcal{G}$. 

\begin{thm}\label{thm:spde}
Let $\warp : \mathcal{D} \mapsto \mathcal{G} := F(\mathcal{D}) \subset \mathbb{R}^d$ be a bijective and differentiable function. Then the deformed random field $X(\psp) := \hat{X}(F^{-1}(\psp))$, where $\hat{X}$ is the solution to \eqref{eq:Dspde} with $\alpha>d/2$, satisfies 
\begin{align}
\left( I - \kappa^{-2} \nabla \cdot H \nabla \right)^{\alpha/2} X = \kappa^{-1} \noise \quad \mbox{in }\mathcal{G}, 
\end{align}
Here $\noise$ is white noise on $\mathcal{G}$, $\nabla$ is the gradient operator on $\mathcal{G}$, $ \kappa(\psp) := \sqrt{\determinant{J[F^{-1}](\psp)}}$, and 
$
H(s) := \damp(\psp)^2 J[F^{-1}]^{-1}(\psp) J[F^{-1}]^{-T}(\psp)$,
where $J[F^{-1}]$ denotes the Jacobian matrix of $F^{-1}$ and $\determinant{J[F^{-1}]}$ denotes the determinant of this matrix.
\end{thm}

The proof of the theorem is given in Appendix \ref{sec:proofs}. The theorem is a generalization of the deformation result stated (but not formally proved) in \citet{lit:lindgren} for $\alpha=2$. 
The reason for formulating the result using a unit dampening constant in \eqref{eq:Dspde} is that the correlation range can be controlled both by $F$ and this constant. So by fixing the dampening constant to 1, we avoid issues with identifiability when later estimating parameters.

The main advantage with defining a non-stationary model in this way is that $X$ has the same marginal variances as $\hat{X}$, which means that the non-stationarity induced by the deformation does not affect the variance. This greatly simplifies the interpretation of the parameters. In order to model random fields of arbitrary marginal variances we will include a scaling, $\tau(\psp)$, as in the original Mat\'ern SPDE. The final model that we will use is therefore,
\begin{align}\label{eq:SPDEG}
\left( I - \kappa^{-2} \nabla \cdot H \nabla \right)^{\alpha/2} (\tau X) = \kappa^{-1} \noise \quad \mbox{in $\mathcal{G}$}.
\end{align}

If we would have $\mathcal{D} = \mathbb{R}^d$,  the variance of the solution would be,
$$
\Var(X(\psp)) = \frac{\Gamma(\alpha-1)}{\Gamma(\alpha)(4\pi)^{d/2}\tau(\psp)^2}.
$$
In this case, if $F(\psp)$ is not varying spatially and $\tau$ is a constant, the covariance of the solution to \eqref{eq:SPDEG} is stationary and given by,
\begin{align}\label{eq:stationary_cov}
\Cov \left( \psp_1, \psp_2 \right) = \sigma^2 C_{\nu} \left( \damp \sqrt{(\psp_1 - \psp_2)^T H^{-1}(\psp_1 - \psp_2)  } \right), \quad \sigma^2 = \frac{\Gamma(\alpha-1)}{\Gamma(\alpha)(4\pi)^{d/2}}\frac{1}{\tau^2},
\end{align}
where $C_{\nu}$ is the Mat\'ern correlation function in \eqref{eq:matern} with smoothness parameter $\nu = \alpha - d/2$~\citep{lit:khristenko}. Thus, the parameters $\kappa$, $\nu$ and $\sigma$ have the same interpretation as for the standard Mat\'ern covariance, and $H$ controls the anisotropy where eigenvectors of $H$ with larger eigenvalues correspond to directions with longer correlation ranges. 
The difference when $F(\psp)$ is allowed to vary in space is that $\kappa(\psp)$ and $H(\psp)$ should be interpreted as local correlation range and anisotropy explaining how the space is locally deformed in a neighborhood of point $\psp$.

For the model to be of practical use, it has to be approximated by a finite dimensional representation. This can be done using the finite element method as in \citet{lit:lindgren}. Here, the solution is approximated by a basis expansion where the coefficients of the expansion are given by a multivariate Gaussian random vector with a sparse precision matrix (for integer values of $\alpha$). The details of the procedure for our model are presented in Appendix \ref{sec:FEM}. For the general fractional case, the method of \citet{lit:bolin3} can be used.

It should be noted that the model of Equation~\eqref{eq:SPDEG} approximated by the finite element method is valid as long as $H$ and $\kappa$ are derived from a function, $F^{-1}$, that is differentiable almost everywhere and which satisfy the inverse function theorem~\citep{lit:spivak}. This means that the model is valid even though a bijective mapping to a deformed space of stationarity does not actually exist. It is sufficient that each local neighbourhood can be deformed to stationarity, see Appendix~\ref{sec:existenceDspace}.

\subsection{Exceedance probability}
We presented an analytical formula for an upper limit of the exceedance probability, Propositon~\ref{thm:exceedanceProb}. 
It was noted that the only quantity that depended on the correlation structure was $\std_{\dot{W}}(t)$.

The model proposed in this section is strictly spatial while a ship is sailing on a dynamically evolving spatial domain. It is possible to model a spatio-temporal process based on the spatial model of this work using, for instance, a space-time separable covariance function or even better a space-time separable covariance function in the lagrangian frame of reference as in~\citet{lit:baxevani3, lit:baxevani4}.
However, in this work we will just consider the purely spatial model. 
Therefore, during a ships journey the realization of the random field is assumed to persist unchanged over time, i.e., we assume that a ship is travelling on a ``freezed'' realization of the spatial random field. Hence, $t$ will in this work only act as a parametrization of the curve over the spatial domain.
For this model we have an explicit formula for $\std^2_{\dot{W}}(t)$.

\begin{cor}
Assume that the curve, $\gamma$, is differentiable with continuous partial derivatives. Then,
\label{cor:stdDeriv} 
\begin{align}
\std^2_{\dot{W}}(t) = \frac{\kappa^2(\psp(t))}{2(\nu-1)}\dot{\psp}(t)^T H^{-1}(\psp(t)) \dot{\psp}(t),
\end{align}
where $\dot{\psp}(t)$ is the vector of partial derivatives of $\psp(t)$ in space, i.e., the derivative of the curve $\gamma$ at point $t$.
\end{cor}
The corollary is based on Proposition~\ref{thm2} and a proof can be found in Appendix \ref{sec:proofs}.
Proposition~\ref{thm2} gives an explicit formula for $\std^2_{\dot{W}}(t)$ for the more general case of any non-stationary model that corresponds to local deformations to a stationary random field.

\section{Inference}\label{sec:estimation}
In order to use the proposed model for applications the model parameters have to be estimated from data. Even before that, a finite dimensional representation of the spatially varying functions for the non-stationary model of Section~\ref{sec:model} has to be chosen. 

\subsection{Representation of spatially varying parameters}
An intuitive approach would be to parameterize the mapping, $\warp$, directly as in the original deformation method~\citep{lit:sampson}. However, our SPDE only depends on the Jacobian matrix of $\warp^{-1}$ and we could therefore parameterize this Jacobian matrix instead. Due to possible identifiability issues, see Appendix~\ref{sec:existenceDspace}, we do not do that either. Instead we choose to parameterize $H(\psp)$ and $\kappa(\psp)$ through an intermediate step, an auxiliary matrix-valued function,
$$\tilde{H}(\psp) := J[\warp^{-1}]^{-1}(\psp) J[\warp^{-1}]^{-T}(\psp).$$ Since the data is on $\R^2$, $\tilde{H}(\psp)$ is a $2\times 2$-matrix. This matrix is defined by $\tilde{H}_{11}(\psp) = \exp(h_1(\psp))$, $\tilde{H}_{22}(\psp) = \exp(h_2(\psp))$, and $\tilde{H}_{12}(\psp) = \tilde{H}_{21}(\psp) = (2S(h_3(\psp))-1)\exp(0.5(h_1(\psp)+h_2(\psp)))$.
Here, $S$ denotes the sigmoid function and $h_1, h_2, h_3$ are some bounded functions. 
Based on this matrix, we can define $\kappa^2(\psp) = |\tilde{H}(\psp)|^{-1/2}$ and $H(\psp) = \kappa(\psp)^2\tilde{H}(\psp)$. The reason for this parameterization is that it ensures that $H(\psp)$ symmetric and positive definite, $\kappa^2 > 0$, and the parameterization is identifiable, i.e., there is a one-to-one relationship between $\kappa, H$ and the functions $h_1, h_2,h_3$. Moreover, it does not put any constraints on $h_1, h_2,h_3$ more than that they should be bounded and real-valued.

We define $h_1,h_2,h_3$ as low-dimensional regressions on cosine functions over the domain of interest: 
\begin{align}
h_i(\psp) = \sum_{p=0}^k \sum_{n=0}^k \beta_{np}^i \cos \left( n\frac{\pi s_1}{C_1} \right) \cos \left( p\frac{\pi s_2}{C_2} \right),\quad i=1,2,3,
\label{eq:cosineSeries}
\end{align}
where $\psp = (s_1,s_2)$ and $C_1, C_2$ denote the width and height of the bounding box of the observational domain. The advantage with this parameterization is that we do not have any restrictions on the coefficients $\beta_{np}^i$ in order to obtain a valid model, this is important since it simplifies the optimization problem later on. 
A truncated cosine series allow for quite flexible modeling of the non-stationarity while ensuring smoothness and a low-dimensional representation of the non-stationarity.
The smoothness is important for two reasons; partly it is a prerequisite for the method introduced in Appendix~\ref{sec:MLlocal}, and partly it is necessary when approximating the parameters as piecewise constant over each triangle in the FEM mesh, see Appendix~\ref{sec:FEM}. 
A similar basis expansion can be used to parameterize $\log(\tau(\psp))$ and/or the marginal mean, $\mu(\psp)$, if spatially varying marginal parameters are needed in the model.

\subsection{Likelihood-based parameter estimation}
Let $y_{jl}$ denote an observed value of $X$ at location $\psp_j\in \gspace$, where $l$ enumerates replicates. We assume the previously proposed Gaussian random field model for the observations but also include a nugget effect to account for spatially small-scale and independent effects. Thus, the assumed model is 
$Y_{jl} = \mu(\psp_j) + X_l(\psp_j) + \epsilon_{jl}$ where $\epsilon_{jl} \sim \N(0, \sigma_{\epsilon}^2)$ are independent, $\mu(\psp_j)$ is the mean value of the latent field at $\psp_j$, and $\{X_l(\psp)\}_l$ are independent replicates of the Gaussian random field model. Using the FEM discretization of the random field, as explained in Appendix~\ref{sec:FEM}, we can write the model as 
\begin{align*}
\bs{Y}_l = \bs{\mu} + A\bs{U}_l + \epsilon_l, \quad U_l \sim \N\left(0 , Q_U^{-1} \right), \quad \epsilon_l \sim \N(0, \sigma_{\epsilon}^2I), \quad l=1,\ldots, L.    
\end{align*}
Here, $L$ is the total number of replicates in the data set, $\bs{Y}_l = (Y_{1l},\ldots, y_{Jl})^T$ is a vector with all observations for replicate $l$, $\bs{\mu}$ is a vector with entries $\mu_j = \mu(\psp_j)$, $A$ is the \textit{observation matrix} with elements $A_{ji} = \phi_i(\psp_j)$ where $\phi_i$ is the $i$:th FEM basis function, $\bs{U}_l$ is the multivariate Gaussian distribution at replicate $l$ of the coefficients of basis functions in the FEM approximation. The parameters of the model, which needs to be estimated, are thus the parameters $\{\beta_{np}^1, \beta_{np}^2, \beta_{np}^3\}_{n,p=0}^k$ and possibly the nugget variance $\sigma_{\epsilon}$, as well as the parameters for a representation of $\tau(\psp)$ and/or $\mu(\psp)$. 
We parameterize $\log(\sigma_{\epsilon}) = \tilde{\sigma}_{\epsilon}$ since $\tilde{\sigma}_{\epsilon}$ has to be positive.

To estimate the proposed model from data we use a maximum likelihood approach. Since both the nugget effect and the random field are Gaussian, the distribution of $\bs{Y}_l$ is Gaussian with mean $\bs{\mu}$ and covariance $\Sigma_Y = AQ_U^{-1}A^T +  \sigma_{\epsilon}^2I$. However, evaluating the log-likelihood of this Gaussian variable directly would require computing $Q_U^{-1}$, which is dense. To be able to take advantage of the sparsity of $Q_U$, we use the fact that the density of $\bs{Y}_l$ can be rewritten as 
\begin{align}
f(\bs{y}) = \left.\frac{f(\bs{u})f(\bs{y}|\bs{u})}{f(\bs{u}|\bs{y})}\right|_{\bs{u} = \bs{u}^*},
\end{align}
where $\bs{u}^*$ is an arbitrary value in the sample space of $\bs{U}$.
Let $\Theta$ denote the set of all parameters, $\bs{y}_l$ an observed realization of $\bs{Y}_l$ and $\bs{y} = \{\bs{y}_l\}_{l=1}^L$. 
Choosing $u^* = 0$ yields the log-likelihood 
\begin{align}
l(\Theta ; \bs{y}) =&
\frac{L}{2}\left(\log\determinant{Q_U} - J\log \determinant{ \sigma_{\epsilon}^2 }
- \log \determinant{Q_{U|Y}} + J\log(2\pi)\right)  \\
 &+ \frac1{2\sigma_{\epsilon}^2}\sum_{l=1}^L(\bs{y}_l-\bs{\mu})^T \left[\frac1{\sigma_{\epsilon}^2}A Q_{U|Y}^{-1}A^T -I\right](\bs{y}_l-\bs{\mu}).
\label{eq:loglik}
\end{align}
To evaluate the log-likelihood, computing the log-determinants and solving linear systems w.r.t. the matrix $Q_{U|Y}$ are needed. Since all matrices involved in the likelihood are sparse, symmetric, and positive definite, this can be handled efficiently using sparse Cholesky factorizations.

Since no analytical formula is available for the parameter values maximizing the likelihood, we perform numerical optimization of the log-likelihood. The log-likelihood depend on the parameters, $\{\beta^i_{np}\}_{npi}$, in a rather complex manner and it is not possible to compute the derivatives of the log-likelihood explicitly. Instead, we use a Quasi-Newton method through the \textit{optimization toolbox} in Matlab to acquire the maximum likelihood estimates.
This can be time consuming since the proposed model can have a large number of parameters. To reduce the computation time, and the risk of finding sub-optimal local maxima, good starting values are needed for the parameters. We obtain such starting values by first computing local parameter estimates as explained in Appendix \ref{sec:MLlocal}. 
The method is based on approximating the parameters as constant in small neighborhoods close to observed locations in space and compute local estimates for each such neighborhood. The estimates are then merged to get a starting value through a least-squares procedure. See Appendix~\ref{sec:MLlocal} for further details.

\section{Results of case study}
\label{sec:results}
In this section we are concerned with the performance of the derived non-stationary SPDE-model of Section~\ref{sec:model} to the application of modeling significant wave height. 
We want to confirm that maximizing the likelihood numerically given the paramaterization of Section~\ref{sec:estimation} is possible and to evaluate how the model perform for the two tasks related to shipping presented in Section~\ref{sec:data}. Furthermore, it is of interest to see if the non-stationary model is needed or an anisotropic Mat\'{e}rn model is sufficient. 
The data used in the results of this section is the ERA-interim dataset presented in Section~\ref{sec:data}.

\subsection{Model fit}
The model presented in Section~\ref{sec:model} was fitted to the data described in Section~\ref{sec:data}. 
Before fitting the model, the data was logarithmized and then standardized marginally by subtracting the sample mean and dividing with the sample standard deviation for the values of each of the spatial locations separately. This was possible since each location had replicate observations. 
The data is partitioned into a training set and a validation set by taking every other day into the training set and the remaining days in the test set. By this partition, each dataset contains $585$ replicates, no closer than 2 days apart.

The proposed model of Section \ref{sec:model} is fitted to the training data using the parametrization of Section \ref{sec:estimation}. In the cosine series expansions of the $h_i$ functions, $k=4$ was used yielding $5\cdot5 = 25$ parameters to be estimated for each function. 
Since the data is marginally standardized, we assume that it have a constant zero mean as well as a unit marginal variance. Hence, except for the coefficients in 
$h_1(\psp), h_2(\psp), h_3(\psp)$, the only other parameters to be estimated are $\std_{\epsilon}^2$ and $\nu$. The model thus required estimating $77$ parameters in total. 

A histogram of the local estimates of the smoothness parameter, $\nu$, using the method of Appendix~\ref{sec:MLlocal} can be seen in Figure~\ref{fig:estSmooth}. All estimated values are between $1$ and $3$ with the median close to $2$, corresponding to $\smooth = 3$. Since the FEM approximation of~\citet{lit:lindgren} only allow for $\smooth$ such that $2\alpha \in \N$, we fix $\alpha=3$ before estimating the other parameters.
\begin{figure}[t]
	\centering
	\includegraphics[width = 0.4\textwidth, keepaspectratio]{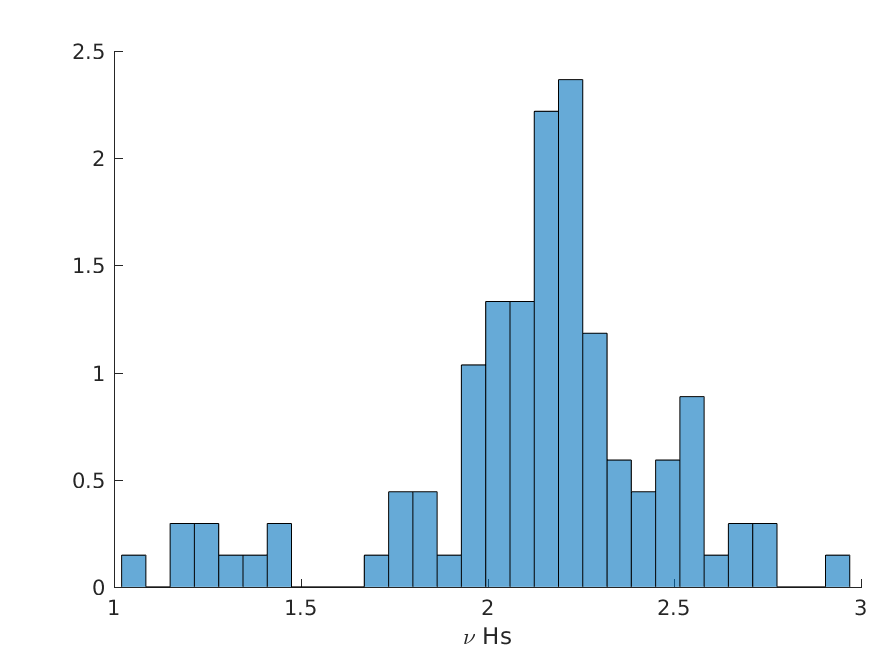}
	\caption{Histogram of estimated smoothness parameters, $\nu$, from the different local estimates. }
	\label{fig:estSmooth}
\end{figure}

\begin{figure}[t]
\begin{minipage}{0.45\linewidth}
\begin{center}
    Data\\
\includegraphics[width = \linewidth]{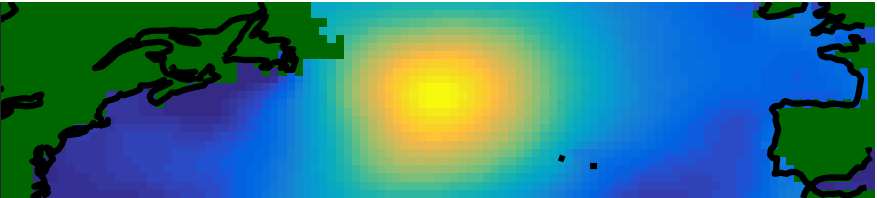}
\includegraphics[width = \linewidth]{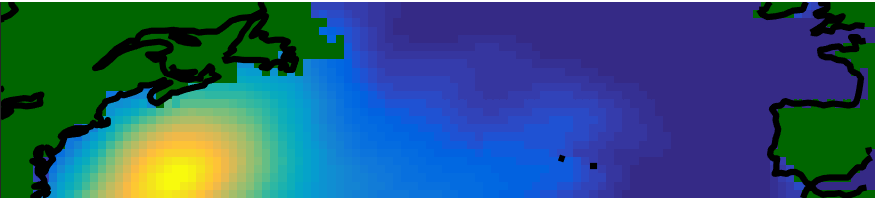}
\includegraphics[width = \linewidth]{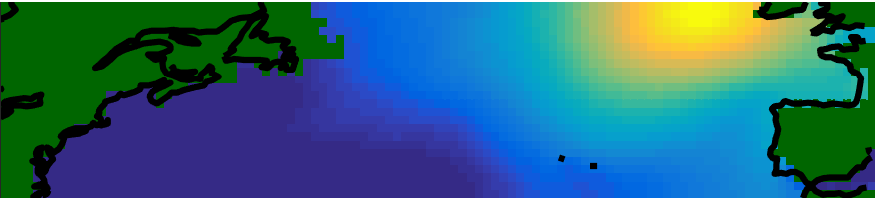}
\end{center}
\end{minipage}
\begin{minipage}{0.45\linewidth}
\begin{center}
    Model\\
\includegraphics[width = \linewidth]{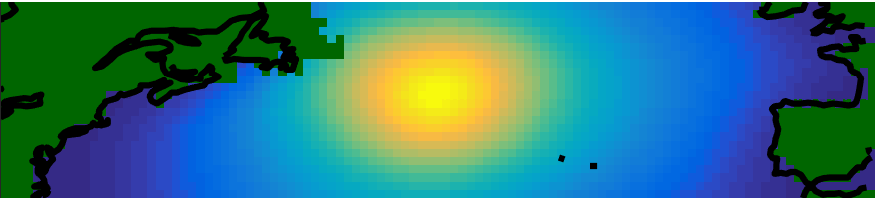}
\includegraphics[width = \linewidth]{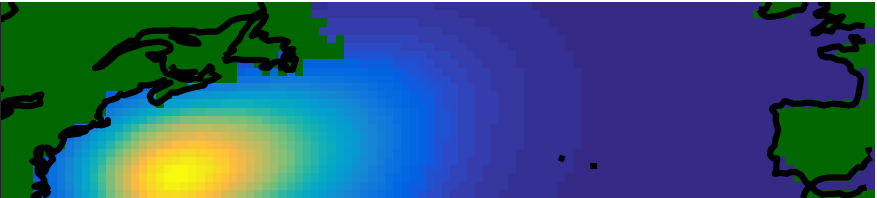}
\includegraphics[width = \linewidth]{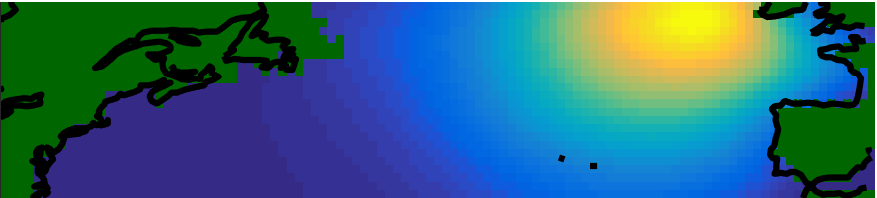}
\end{center}
\end{minipage}
\begin{minipage}{0.08\linewidth}
\begin{center}
\includegraphics[width = \linewidth]{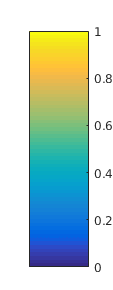}
\end{center}
\end{minipage}
	\caption{Empirical (left) and model (right) correlations between reference points and all other points in the spatial domain. The reference point is 
	in the middle of the north Atlantic ocean (top row), close to the 
	west coast of USA (middle row), and to the west coast of central Europe (bottom row). 
	}
	\label{fig:correlation}
\end{figure}

Three examples of correlation functions between one point in the spatial domain with all other points can be seen in Figure \ref{fig:correlation}. The figure shows both the empirical correlation function from data and the computed correlation function from the estimated model.
As can be seen in Figure~\ref{fig:correlation}, the correlation structure is well captured by the model.
Another way of visualizing the non-stationarity of the estimated model is to show the estimated deformation.
Figure \ref{fig:mapping} show how a rectangular grid of the north Atlantic (the $\gspace$-space) has been mapped to $\dspace$ acquired from the parameter estimation of the non-stationary model. 
In the figure one sees that distances close to the coasts are elongated compared to the middle of the Atlantic ocean. This means that correlation drops of quicker with distance close to the coasts. Such a behaviour makes sense due to the various effects occurring on the interface between land and ocean, causing spatially (and temporally) less large scale dependence. 
The observations are placed on a grid in the longitude/latitude coordinate system. This will make the distance in the longitudal direction between points further north larger than for points in the south of the domain. If the field was stationary on the sphere, this would correspond to an elongation in the x-direction in $\dspace$ as the y-value decreases. Such an effect is not clearly visible in the figures, likely because the actual random field is not stationary even in the spherical frame of reference.

It should be noted that the grid has folded in a few positions close to the coast, suggesting that the estimated $\dspace$ is not well-defined. 
The folding is hard to avoid, see Appendix \ref{sec:dspace}. However, for our model, folding is not an issue and the model is still valid---this is not the case for the original deformations method~\citep{lit:sampson}. 
Moreover, even with folding, Figure~\ref{fig:mapping} still gives an interpretation of the non-stationarity and anisotropy since it is possible to see the local elongations and rotations.

\begin{figure}[t]
	\centering
	\begin{subfigure}{0.39\textwidth}
		\centering
		\includegraphics[width = \textwidth]{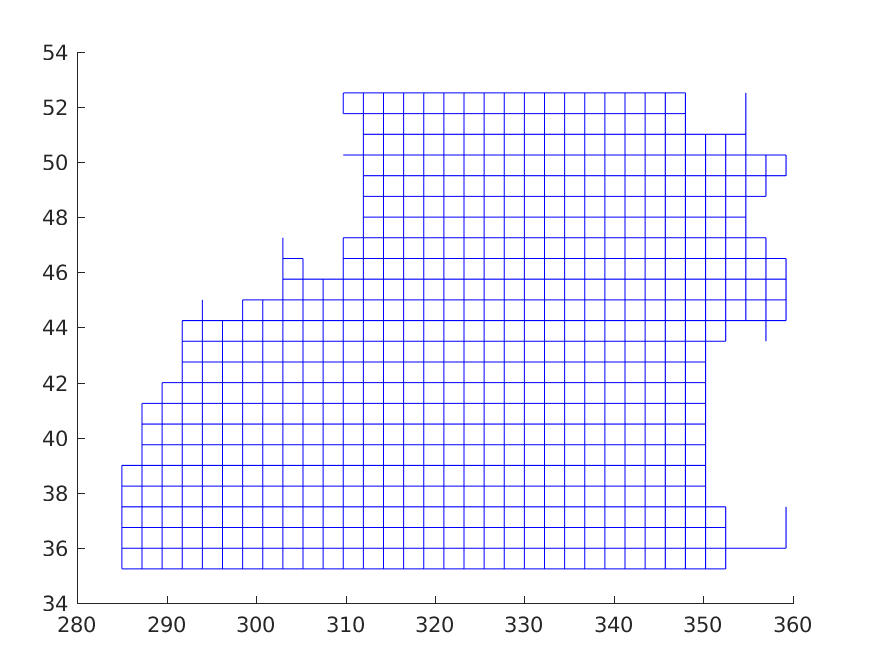}
	\end{subfigure}
	\begin{subfigure}{0.39\textwidth}
		\centering
		\includegraphics[width = \textwidth]{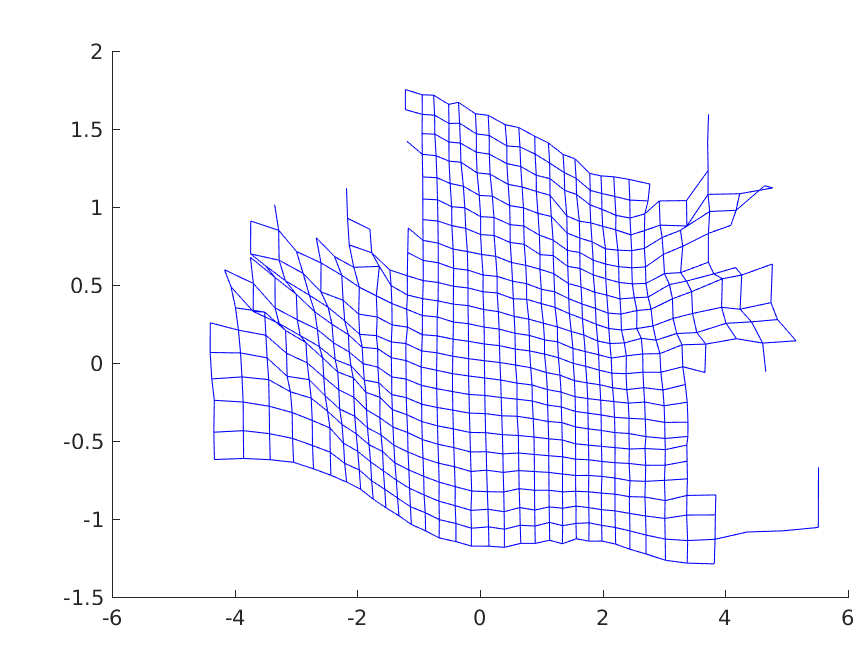}
	\end{subfigure}	
	\caption{Mapping of a grid in $\gspace$ (left) to $\dspace$ (right). The grid points outside of the ocean have been removed.}
	\label{fig:mapping}
\end{figure}

 \subsection{Comparison stationarity/non-stationarity}
The results of the previous subsection suggest that the non-stationary model is needed for the data. However, a relevant question is whether the fitted non-stationary model explains the distribution of the data significantly better than a fitted stationary model. To answer this, we also fitted a stationary anisotropic model to the data.
 Since the likelihood function is explicitly available and the two models are nested, we can perform a likelihood-ratio test between them. Here, the stationary model is considered to be the baseline model so the non-stationary model represent the more flexible alternative.
 The statistic we consider is the logarithm of the ratio between the maximum likelihood values of the two models, 
 \begin{align*}
     \lambda = \log \frac{\sup_{\theta \in \Theta_0}L(\theta ; y) }{ \sup_{\theta \in \Theta}L(\theta ; y)} = l_{\text{stationary}} - l_{\text{non-stationary}}.
 \end{align*}
 The null hypothesis is that the stationary model explains the data as well as the non-stationary one. The parameter space $\Theta_0$ is the subspace of parameter values which yields a stationary model ($k=0$ in Equation~\eqref{eq:cosineSeries}) and $\Theta$ is the full parameter space of the non-stationary model (in our case $k=4$ in Equation~\eqref{eq:cosineSeries}). 
 The null hypothesis is rejected if $\lambda < c$ for some critical value $c$. We cannot compute $c$ for a given significance level explicitly. 
 However, since the number of realizations of the spatial observations are rather big, $585$, it is reasonable to consider asymptotic results. 
 Using that the distribution of $-2\lambda$ converges to a $\chi^2_{df}$-distribution \citep[Theorem 10.3.3]{lit:casella} it is possible to compute $c$.
 Here, the degrees of freedom, $df$, equals the difference between the number of parameters in the two models, in our case $df = 72$. This means that $c = -\frac{1}{2}F^{-1}_{\chi^2_{72}}(1-\alpha)$ where $\alpha$ is a chosen significance level and $F^{-1}_{\chi^2_{df}}$ is the quantile function of the chi-square distribution.
 Choosing for instance $\alpha = 10^{-4}$ yields $c = -62.688$.
 The observed value of $\lambda$ is $2.652\cdot 10^6 - 2.907\cdot 10^6 = -2.55\cdot 10^5$, clearly $\lambda \ll -62.688$ and the non-stationary model is significantly better according to the test even with an extremely low significance level. Because of the large likelihood difference, the non-stationary model is also preferable according to the Akaike information criterion \citep{lit:akaike}.

\subsection{Fatigue damage accumulated by a vessel}
\label{sec:route}
For model validation we will compare the empirical distribution of accumulated fatigue damage between the data and simulated data from the fitted model for the route shown in Figure~\ref{fig:atlanticRoute}, which is a scenario similar to the one in \citet{lit:podgorski}. The continuous route is approximated by line segments between $100$ point locations (evenly spaced in geodesic distance). 

\begin{figure}[t]
    \centering
    \includegraphics[width=0.6\textwidth]{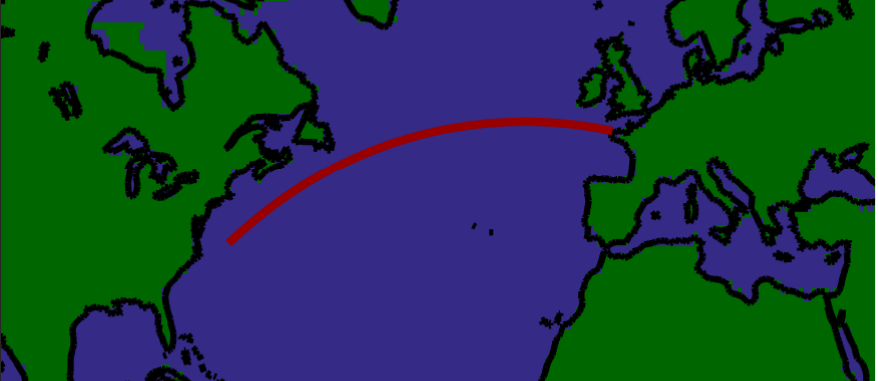}
    \caption{Atlantic route.}
    \label{fig:atlanticRoute}
\end{figure}
We set the ships speed to a fixed value of $10$ m/s which yields a sailing duration of $149.69$ hours or equivalently $6.23$ days. The heading of the ship, in one of the 100 locations en route, is approximated as the mean between the direction acquired from the two connecting line segments.
We set the constants specific to the ship as in~\citet{lit:podgorski, lit:mao}, i.e. $C = 20, \beta = 3,$ and $\gamma = 10^{12.73}$. 
In order to use Equation~\eqref{eq:fatigueDamage} we also need the propagating waves angle in comparison with the ship.
In \citet{lit:podgorski} the mean wave propagation in space was estimated for each month of a year. We use those propagation estimates in order to acquire the instantaneous angle between the ships heading and the expected direction of the traveling waves during the month of April.

Our model is purely spatial and the temporal resolution of the data is low (one measurement per day). We want to compare the empirical distribution of the accumulated damage using data with simulations from our model. Therefore, as mentioned in Section~\ref{sec:model}, we assume that the sea state is constant in time but spatially varying during the ship's journey. Finally, to acquire the accumulated damage, the numerical integration of Equation~\eqref{eq:fatigueDamageInt} is computed over the 100 point locations on the route. 

We use the validation data to compute the empirical fatigue damage distribution by calculating total fatigue damage of the ships journey for each of the 585 replicates in the dataset. We then simulate new $H_s$ data sets of equal size (585 independent replicates) from the model and compute the correponding fatigue damage distribution. 
Quantile-quantile (QQ) plots of 200 simulated fatigue damage distributions against the empirical distribution from the data can be seen in Figure \ref{fig:fatigueQQ}. In the figure, three models are compared both traveling from USA to Europe and vice versa. The direction of the route affects the fatigue damage since the waves will have a different angle of attack when the ship is travelling in one direction compared to the opposite. 
The three models compared are the non-stationary model, the stationary model, and an independent model where each of the 100 point locations on the route are modeled independently using the marginal normal distributions of $\log H_s$.
Clearly, a spatial model is needed since the variability is underestimated in the independent model, which does not resemble the distribution of the data at all. 
The stationary and non-stationary models both seem to capture the bulk of the accumulated fatigue damage probability distribution well. However, in the lower tails the non-stationary model show a better fit. The difference is most clearly seen on the route from Europe to USA which also seem to yield a greater amount of damage. 
 
\begin{figure}[t]
\centering
\begin{subfigure}{0.32\textwidth}
\centering
Non-stationary
\end{subfigure}
\begin{subfigure}{0.32\textwidth}
\centering
Stationary
\end{subfigure}
\begin{subfigure}{0.32\textwidth}
\centering
Independent
\end{subfigure}\\
\hspace{0.2cm}
\begin{subfigure}{0.32\textwidth}
\centering
\includegraphics[width=\textwidth]{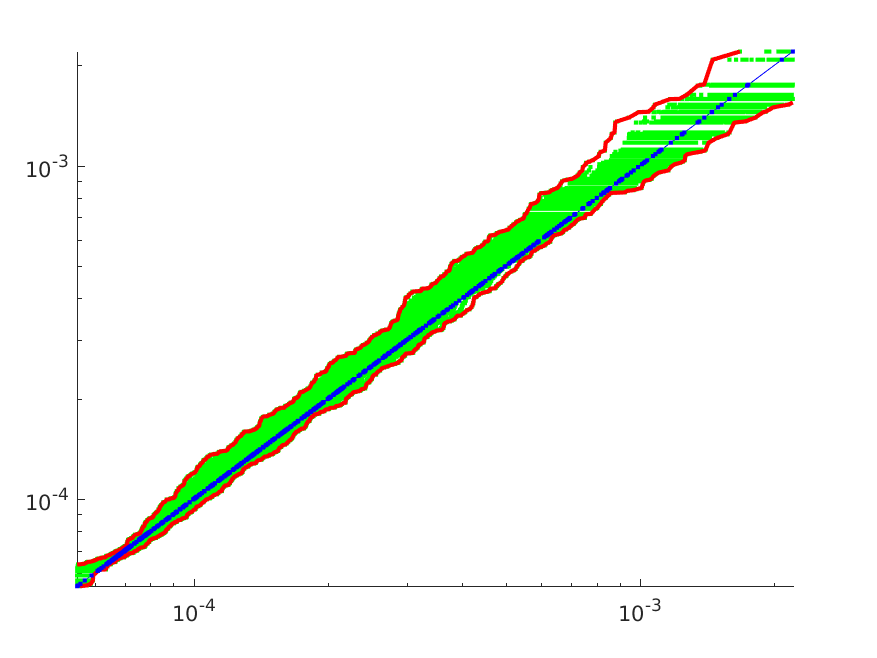}
\end{subfigure}
\begin{subfigure}{0.32\textwidth}
\centering
\includegraphics[width=\textwidth]{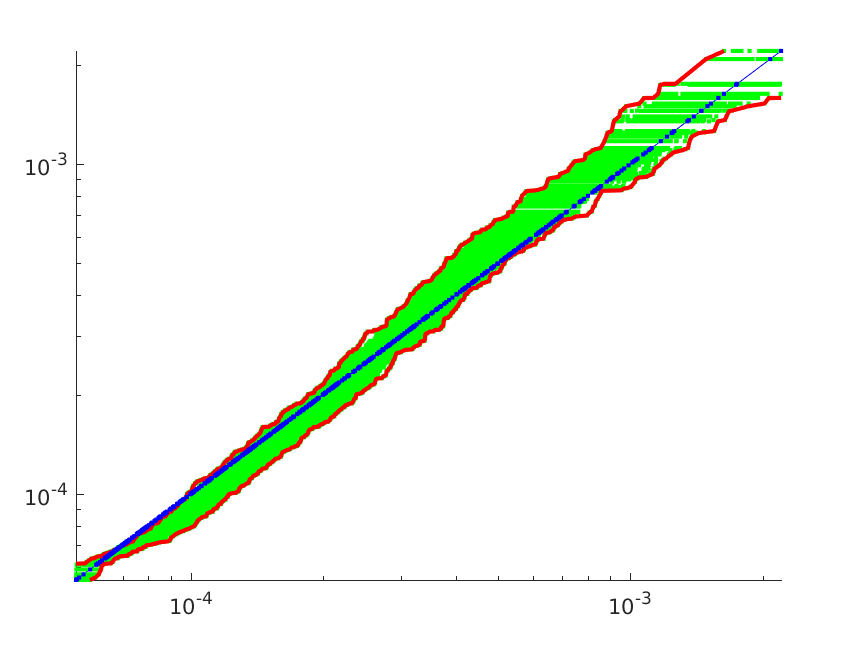}
\end{subfigure}  
\begin{subfigure}{0.32\textwidth}
\centering
\includegraphics[width=\textwidth]{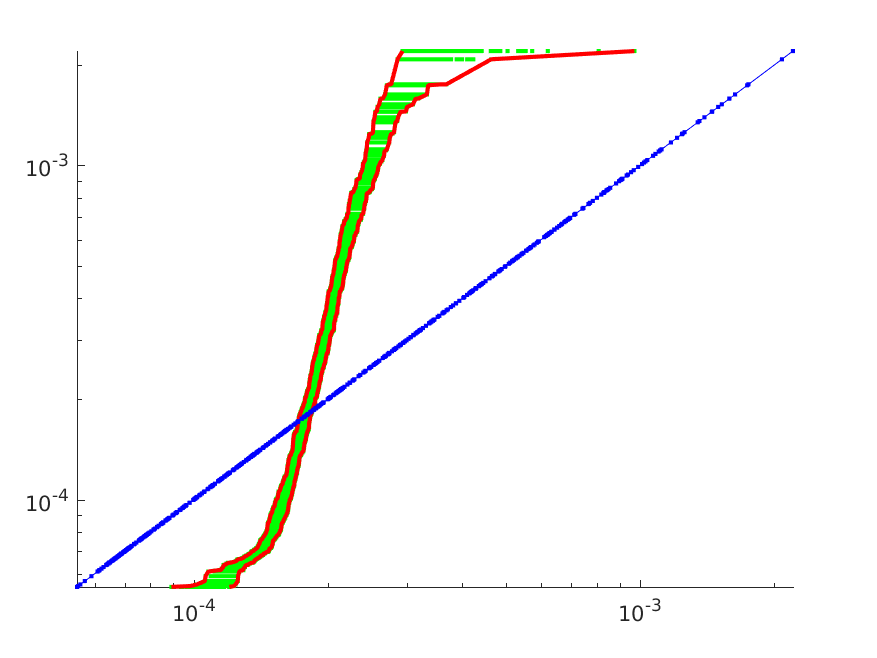}
\end{subfigure}\\
\begin{subfigure}{0.32\textwidth}
\centering
\includegraphics[width=\textwidth]{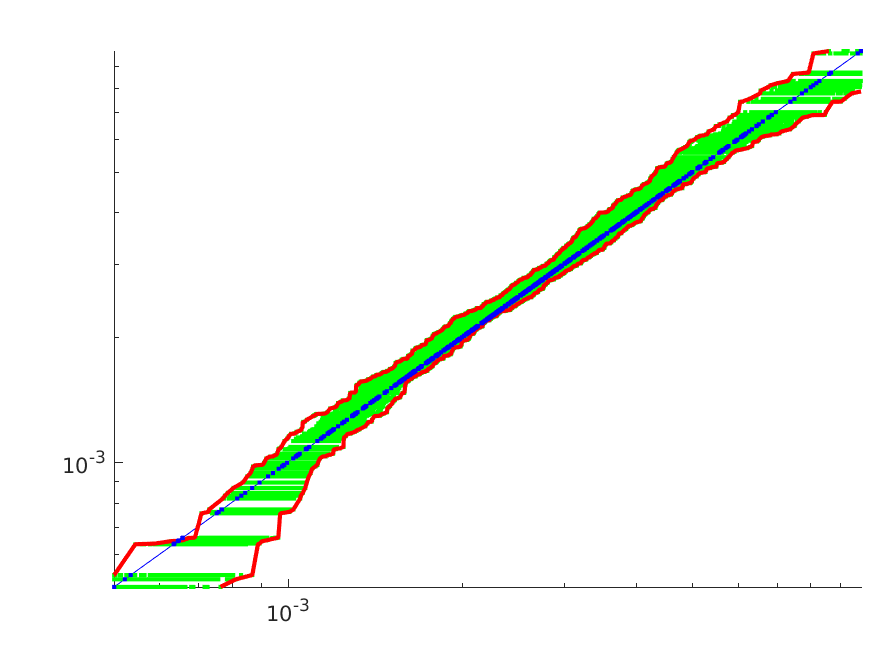}
\end{subfigure}
\begin{subfigure}{0.32\textwidth}
\centering
\includegraphics[width=\textwidth]{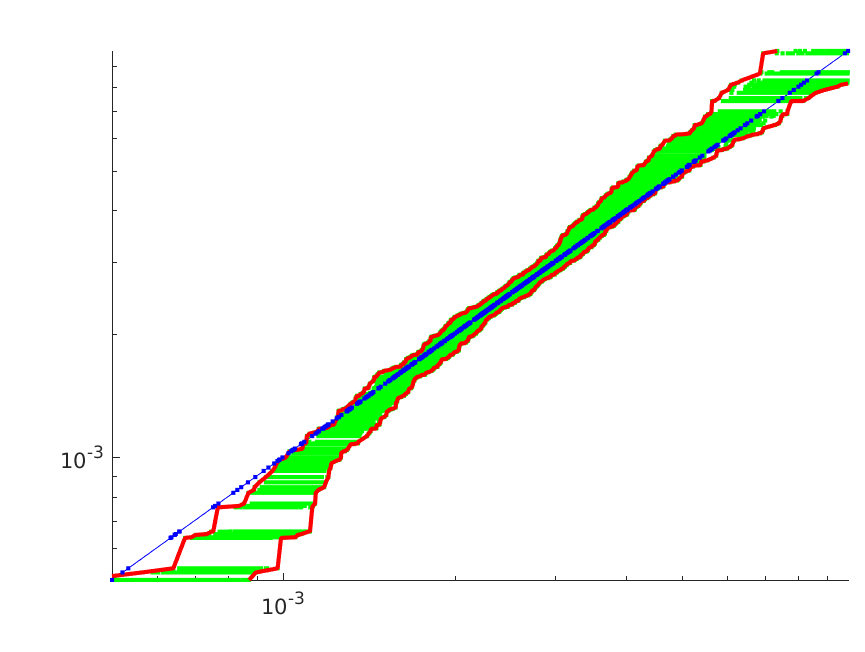}
\end{subfigure}  
\begin{subfigure}{0.32\textwidth}
\centering
\includegraphics[width=\textwidth]{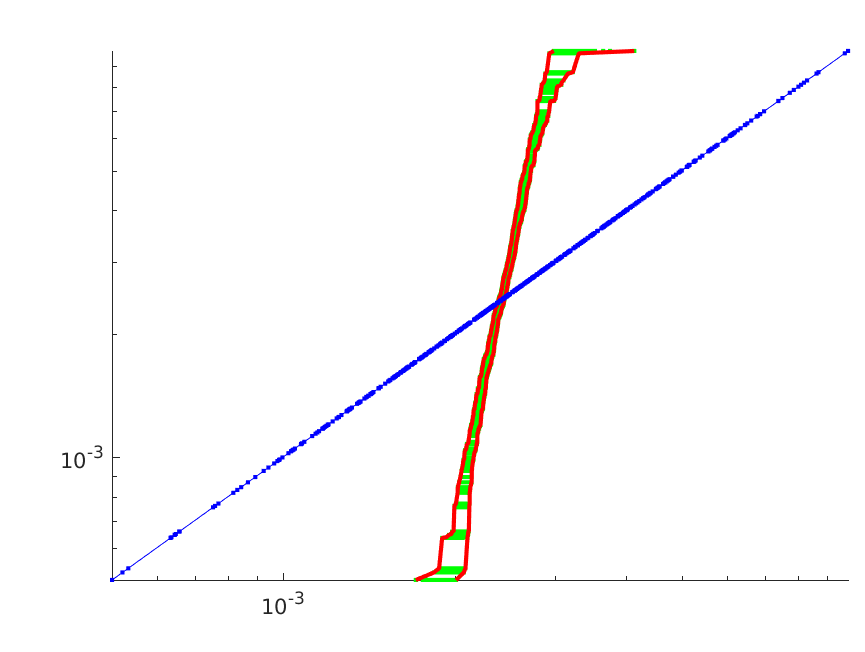}
\end{subfigure}  

\caption{QQ-plots of fatigue damage for data (blue) and 200 simulations (green). The envelopes of the simulations are drawn in red.
 Top row shows damage when traveling the route eastwards and bottom row show damage when traveling westwards. The columns from left to right show the results for the non-stationary, stationary, and independent model respectively.}
 \label{fig:fatigueQQ}
 \end{figure}

 \subsection{Extreme loads}
 \label{sec:extremeLoads}
 Section~\ref{sec:data} addressed the need for analyzing the risk of encountering extreme significant wave heights during a journey.
 Also in this applications we will consider the a ship planned to travel the route of Figure~\ref{fig:atlanticRoute}. 
 We compare the exceedance probability estimated from data with those estimated from the model.
 Estimated model exceedance probabilities are both acquired by Monte-Carlo simulations and by the explicit upper bound calculated using Proposition~\ref{thm:exceedanceProb} and Corollary~\ref{cor:stdDeriv}. 
 
 The exceedance probability of a given threshold for the significant wave height is estimated from data by counting the fraction of replicates for which the threshold was exceeded on the considered route. 
 The Monte-Carlo simulations from the model are computed in the same way for 200 simulations from the fitted model---each one with 585 replicates. 
 Since we are considering a purely spatial model there will be no difference in exceedance probability depending on the direction traveled. 
 
 To compute the theoretical bound, we approximate the integral of Proposition~\ref{thm:exceedanceProb} by a sum over the 100 points on the route. The derivatives of the mean and standard deviation fields are approximated numerically using the difference between values of neighbouring pixels on route.
 
 Figure~\ref{fig:exceedance} show the estimated exceedance probabilities for a range of thresholds between 2 to 12 meters. The figure compare the empirical exceedance probability from the data with both the stationary and non-stationary model. 
 The non-stationary model has an envelope that encapsulates the data completely, suggesting that the model is sufficient for modeling exceedance probabilities. For the stationary model, the exceedance probability is below the envelope in the lower range of threshold values (up until about $5$ meters), meaning that the model overestimates the exceedance probability of small values. 
 The upper bound seems to perform well as an approximation of the exceedance probability for values above 5 meters for the non-stationary model and values above 6 meters for the stationary model, as suggested by the location of the black line in comparison with the bulk of the green lines. 
 For the very high threshold values, which are the exceedance probabilities which are typically of interest, the empirical estimates perform poorly due to lack of sufficient data. This is the range where the upper bound is expected to behave as the real exceedance probability. However, it is hard to validate this since the reference (the blue line) is uncertain in this region.

  \begin{figure}[t]
 \centering
 \begin{subfigure}{0.48\textwidth}
 \includegraphics[width=\textwidth]{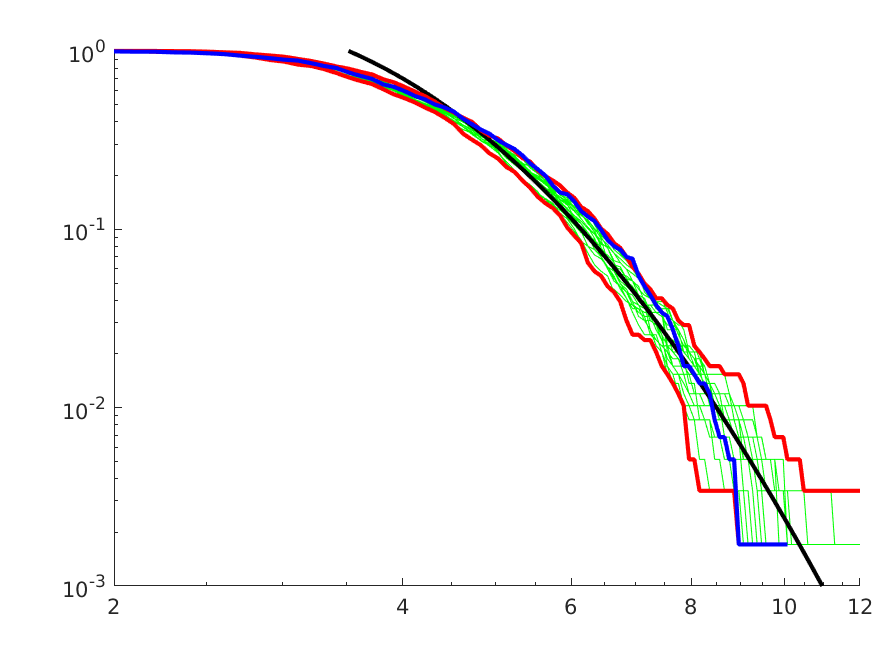}
 \caption{Non-stationary model}
 \end{subfigure}
 \begin{subfigure}{0.48\textwidth}
 \includegraphics[width=\textwidth]{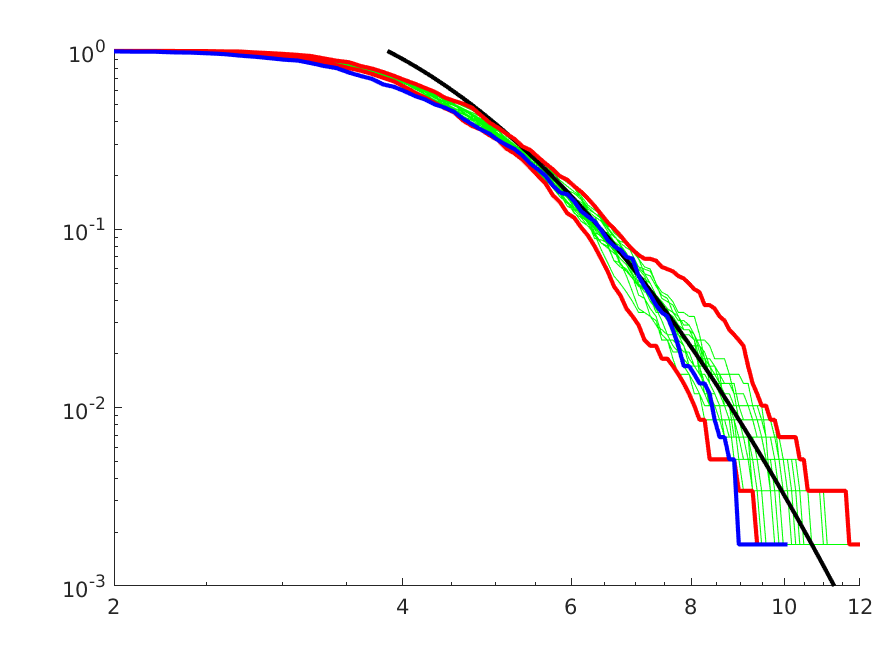}
 \caption{Stationary model}
 \end{subfigure}
 \caption{Comparison of the empirical exceedance probability from data (blue), the corresponding empirical exceedance probabilities from 200 model simulations (green), as well as the upper bound from the model (black).
 The y-axis show the exceedance probability for the corresponding threshold on the x-axis. The red lines show the envelope of the 200 simulations.  }
 \label{fig:exceedance}
 \end{figure}

\section{Discussion}
\label{sec:discussion}
We have developed a non-stationary and/or anisotropic Gaussian random field model for significant wave heights.
The model is formulated using the SPDE representation of Gaussian Mat\'ern fields in combination with the deformation method of \citet{lit:sampson}. Our main result shows that a deformed Mat\'ern SPDE model can be represented as a solution to another fractional SPDE on the deformed space. 

For the special case of $\alpha=2$, \citet{lit:fuglstad, lit:fuglstad2} considered the related model $$(\tilde{\kappa}^2 - \nabla \cdot \tilde{H} \nabla) u = \noise,$$
where the functions $\tilde{H}(\psp)$ and $\tilde{\kappa}(\psp)$ could be chosen independently of each other. Because of this, the marginal variance of the solution is affected by the choice of $\tilde{H}(\psp)$ and $\tilde{\kappa}(\psp)$. In the stationary case, a direct consequence of comparing the covariance function of our model to the covariance function derived in~\citet{lit:fuglstad2} shows that the two models are equivalent (when using $\alpha = 2$ in our model). However, the important difference between the models is that our parametrization allow for controlling the variance of the process independently of $\kappa$ and $H$. For the stationary case, this is not crucial, but when changing to non-stationary models it is important for the interpretability of the parameters and simplifies estimation. In particular, it guarantees that the model parameters are identifiable, which can be a difficult thing to check in more general models. It also solves the problem stated in~\citep{lit:fuglstad2} of identifying the marginal variance for a non-stationary SPDE model. 

Another important advantage with the proposed model is that the deformation method allow for theoretical derivations of results that are typically difficult to obtain for non-stationary models. One example of such is the upper bound on exceedance probabilities which we derived. 
%

Just as for the more basic SPDE-based models, Gaussian Markov random field approximations of the model can be obtained using the finite element method. This allows for computationally efficient inference and simulation, also for complex spatial domains and spatially irregular measurements. Thus, the proposed model will likely be useful also in many other applications not related to wave modelling.



Significant wave height data from the ERA-interim reanalysis dataset was used to fit and validate the model.
Results showed that both the exceedance probability of significant wave height for a specified shipping route, as well as the 
accumulated fatigue damage distribution, are well explained by the proposed model. 
The simpler stationary model performed well on the accumulated fatigue damage but underestimated small exceedance probabilities.

There are many possible extensions to the proposed model which will be investigated in future work.
One such extension is to model multivariate random fields similar to~\citep{lit:hu, lit:bolin}. Such models would allow for joint modelling of significant wave height and mean zero crossing period, which for example would allow for a more accurate fatigue damage model. 
%
Another important aspect of the wave state data is that it is evolving in time. Therefore, another natural continuation is to derive a suitable spatio-temporal extension of the model. 


\section{Acknowledgements}
This work has been supported in part by the Swedish Research Council under grant No. 2016-04187. We would like to thank the European Centre for Medium-range Weather Forecast for the development of the ERA-Interim data set and for making it publicly available. 
The data used was the ERA-Interim reanalysis dataset, Copernicus Climate Change Service (C3S) (accessed September 2018), available from~\url{https://www.ecmwf.int/en/forecasts/datasets/archive-datasets/reanalysis-datasets/era-interim}.


\appendix

\section{Finite element approximation}
\label{sec:FEM}

To obtain a finite dimensional approximation of the model in Section \ref{sec:model} we assume that $\gspace$ is a bounded and polygonal domain. The solution of the SPDE is then approximated by a basis expansion, $X(\psp) = \sum_{i=1}^N U_i \phi_i$, where $\{\phi_i\}$ are piecewise linear basis functions induced by a triangular mesh on the spatial domain, see Figure~\ref{fig:FEMbasis}. The distribution of the coefficients $\bs{U} = (U_1,\ldots, U_N)^T$ are computed by replacing the infinite-dimensional test space of the weak formulation of the SPDE with the finite-dimensional subspace spanned by the basis functions,~$\{\phi_i\}$. 

Furthermore, when computing the inner products in the FEM problem, the spatially varying parameters, $\kappa(\psp), H(\psp)$, are in the implementation of Section~\ref{sec:results} approximated as constant over each triangle. The value over a triangle is taken as the value in the mid-point of the triangle. This approximation will make the inner products of the FEM possible to calculate explicitly, it also makes Equation~\eqref{eq:SPDEG} equivalent with
\begin{align}
    \left( \kappa^{2/\alpha}\mathcal{L} \right)^{\alpha/2}X = \noise,
\end{align}
since $\kappa$ and $\mathcal{L}$ will commute. Here, $\mathcal{L} := \eye - \kappa^{-2}\nabla \cdot H \nabla$.
Assuming the parameters to be constant over each triangle should not affect the solution significantly since $\kappa(\psp), H(\psp)$ are varying slowly in comparison with the spatial extent of a triangle in the mesh. 
 
For $\alpha=2$, the distribution of $\bs{U}$ is computed by solving the following system of linear equations
\begin{align}
\left\{\sum_j \langle \damp \mathcal{L} \tau \phi_j , \phi_i \rangle U_j, i=1,\ldots N\right\}
&\overset{d}{=} \left\{\langle  \noise, \phi_i \rangle , i=1,\ldots N \right\}.
\label{eq:FEMformulation}
\end{align}

Using Stokes theorem in combination with the Dirichlet boundary conditions, the inner products on the left hand side can be written as
\begin{align}
\langle \mathcal{L}  \tau \phi_j , \phi \rangle &= 
 \langle \damp \tau \phi_j , \phi \rangle - \langle \damp^{-1}\nabla \cdot H \nabla (\tau \phi_j) , \phi \rangle \\
&= \langle \damp \tau \phi_j , \phi \rangle + \langle  H \nabla (\tau \phi_j) , \nabla (\damp^{-1}\phi) \rangle.
\end{align}

The left-hand side of \eqref{eq:FEMformulation} can therefore be written as 
$K U = (B + G)U$, where $B_{ij} = \langle \damp \tau \phi_j, \phi_i \rangle$ and $G_{ij} = \langle  H \nabla (\tau \phi_j), \nabla \left( \damp^{-1} \phi_i \right) \rangle$. In order to obtain a Markov random field, the right-hand side is approximated by a centered multivariate Gaussian random variable $W$ with diagonal covariance matrix $C$, with elements $C_{ii} = \langle 1, \phi_i \rangle$~\citep{lit:lindgren}. 
Solving $KU = W$, yields $U \sim \mathbb{N}(0, Q^{-1})$, with the sparse precision matrix $Q =K^TC^{-1}K$. 



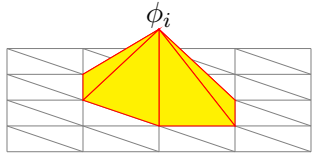
\begin{figure}[t]
	\centering
	\tdplotsetmaincoords{70}{0}
	\begin{tikzpicture}[scale=1, tdplot_main_coords]
	
	\draw[step=1cm,gray,very thin] (-2, 2,0) -- (2,2,0);
	\draw[step=1cm,gray,very thin] (-2, 1,0) -- (2,1,0);
	\draw[step=1cm,gray,very thin] (-2, 0,0) -- (2,0,0);
	\draw[step=1cm,gray,very thin] (-2, -1,0) -- (2,-1,0);
	\draw[step=1cm,gray,very thin] (-2, -2,0) -- (2,-2,0);
	
	\draw[step=1cm,gray,very thin] (2, -2,0) -- (2,2,0);
	\draw[step=1cm,gray,very thin] (1, -2,0) -- (1, 2,0);
	\draw[step=1cm,gray,very thin] (0, -2,0) -- (0, 2,0);
	\draw[step=1cm,gray,very thin] (-1, -2,0) -- (-1, 2,0);
	\draw[step=1cm,gray,very thin] (-2, -2,0) -- (-2,2,0);			
	
	\draw[step=1cm,gray,very thin] (-2, 2,0) -- (2,-2,0);
	\draw[step=1cm,gray,very thin] (-2, 1,0) -- (1,-2,0);
	\draw[step=1cm,gray,very thin] (-1, 2,0) -- (2,-1,0);
	\draw[step=1cm,gray,very thin] (0, 2,0) -- (2,0,0);
	\draw[step=1cm,gray,very thin] (-2, 0,0) -- (0,-2,0);
	\draw[step=1cm,gray,very thin] (1, 2,0) -- (2,1,0);
	\draw[step=1cm,gray,very thin] (-2, -1,0) -- (-1,-2,0);
	
	\draw[red, fill=yellow] (-1,1,0) -- (0,0,1) --
	(0,1,0) -- cycle;
	\draw[red, fill=yellow] (0,1,0) -- (1,0,0) --
	(0,0,1) -- cycle;
	\draw[red, fill=yellow] (1,0,0) -- (0,0,1) --
	(1,-1,0) -- cycle;
	\draw[red, fill=yellow] (1,-1,0) -- (0,0,1) --
	(0,-1,0) -- cycle;
	\draw[red, fill=yellow] (0,-1,0) -- (0,0,1) --
	(-1,0,0) -- cycle;							
	\draw[red, fill=yellow] (-1,0,0) -- (0,0,1) --
	(-1,1,0) -- cycle;   
	
	\node (phi) at (0,0,1.2) {$\phi_i$};
	
	\end{tikzpicture}	
	\caption{Example of a test function for node $i$ on a 2-dimensional mesh.}
	\label{fig:FEMbasis}
\end{figure}

Equation \eqref{eq:FEMformulation} provides an approximation of the solution when $\smooth = 2$. For models where $\smooth$ is an even integer \citet{lit:lindgren} showed how a solution can be expressed by recursively applying $\damp^{2/\alpha}\mathcal{L}$ several times. Also for $\smooth$ being an odd integer, a similar construct could be achieved by considering the least squares solution to $\damp^{1/2}\mathcal{L}^{1/2}X = \noise$ and then considering recursive solutions to this. This is possible also for our anisotropic model if $\kappa$ is piecewise constant.
Hence, for any integer-valued $\alpha$, this scheme yields the precision matrices $Q^{(1)} := K$ for $\alpha = 1$, $Q^{(2)} := KC^{-1}K$ for $\alpha = 2$, and integer precision matrices recursively being defined as $Q^{(\alpha)} := KC^{-1}Q^{(\alpha-2)}$ for $\alpha \in \{ 3, 4, \ldots \}$.
These matrices become less sparse as $\smooth$ is increased. However, the computational complexity of using the approximation for inference will be $\mathcal{O}(N^{3/2})$ but with a larger constant due to less sparsity for larger $\smooth$. This should be compared to the standard cubic increase in computational cost for general covariance based methods~\citep{lit:rue}.

\subsection{Boundary values and meshing}
The use of boundary conditions for the bounded domain in the FEM approximation will affect the solution to the SPDE such that behaviour close to the boundary deviates from that of the corresponding model on an unbounded domain. However, the boundary conditions will have a negligible effect in spatial regions that are sufficiently far from the boundary due to the positive dampening, $\kappa$~\citep{lit:khristenko}. \citet{lit:lindgren} used this fact and simply extended the finite element mesh far enough outside of the spatial domain of interest such that the solution in the domain of interest was practically unaffected by the boundary conditions. For a Gaussian random field with a Mat\'{e}rn covariance function the extension distance is often chosen as $r$ or $2r$ where $r = \frac{\sqrt{8\nu}}{\kappa}$ denotes the practical correlation range.

Here there are two effects that have to be balanced: In order to reduce the computational cost, the number of triangles should be as small as possible. At the same time, the mesh diameter (the longest edge in the triangulation) has to be small enough in order to get a good FEM approximation of the true solution. We let the mesh diameter be at most a fifth of the local correlation range at the location of the triangle. However, when $\damp$ and $H$ are spatially varying, we do not have a fixed correlation range. It would be a waste of computational resources to extend the mesh based on the largest local correlation range while keeping the mesh fine enough everywhere to resolve the smallest local correlation range. By applying the barrier method of \citet{lit:bakka} we get around this issue enforcing the smallest local correlation range in all of the extended region outside of the observational domain. Hence, the extension distance only has to be two times the smallest local correlation range on the observational domain. This effectively minimizes the number of triangles needed in the extended domain while keeping a good FEM approximation of the true solution. 
Setting the extended triangles to the smallest correlation range is achieved by setting the Jacobian matrix $J[\warp^{-1}](\psp) = \frac{\sqrt{8\nu}}{r_{\text{min}}} I$ for all triangles that are exclusively in the exterior of $\gspace$. Here, $r_{\text{min}}$ denotes the minimum local correlation range in the interior. 
In Figure~\ref{fig:mesh}, the blue triangles are part of the spatial domain of interest and the gray triangles are part of the extension. 


\begin{figure}[t]
	\centering
	\includegraphics[width=0.5\textwidth]{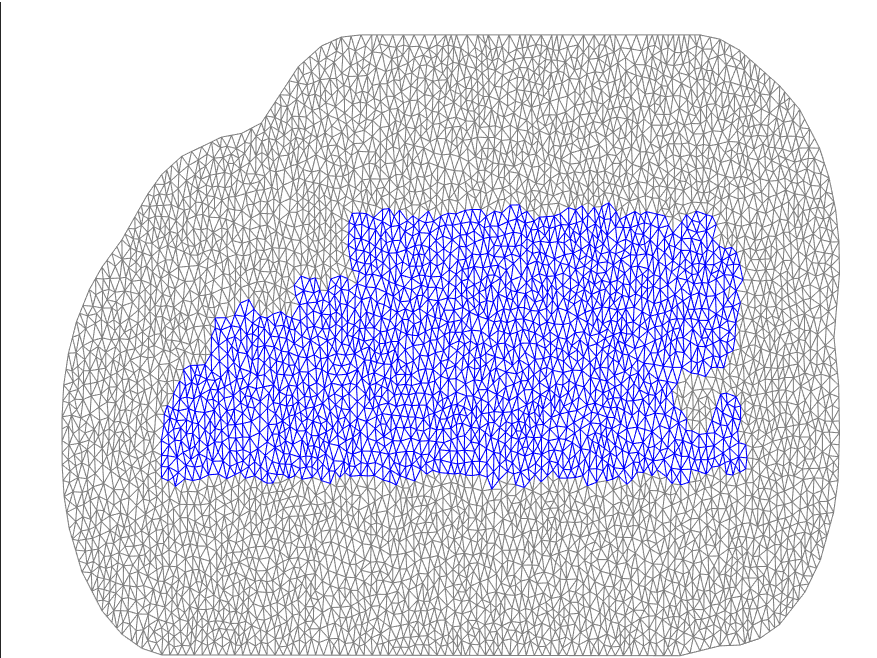}
	\caption{The mesh of the north Atlantic used in Section~\ref{sec:results}. Triangles of the extension are shown in gray and triangles from the interior in blue.}
	\label{fig:mesh}
\end{figure}

\section{Local parameter estimation}
\label{sec:MLlocal}
To obtain starting values for the parameter estimation, we first estimate the parameter values of the model of Section~\ref{sec:model} locally in small spatial regions by approximating the parameters as constant in neighborhoods of $3\times 3$ pixels, see Figure~\ref{fig:localEstimation}. 
For each local estimate, we first estimated the marginal mean and variance for each location in the neighborhood using the regular sample mean and sample variance. We then compute the parameters of the constant $H$ matrix by numerical optimization of the log-likelihood when assuming that the covariance function in the local neighborhood is defined as in Equation~\eqref{eq:stationary_cov}. 
It should be noted that we do not estimate local parameters for all points but for a selection of points spread out evenly in the domain. 

The local estimates are merged to get spatially varying parameters according to the series of cosines specified in Section~\ref{sec:estimation}. This is accomplished by computing the least squares solution of the values $\{\beta_{np}^i\}$ based on the local estimates.
The merged values are then used to acquire good initial values for the numerical maximization of the likelihood function as described in Section~\ref{sec:estimation}.


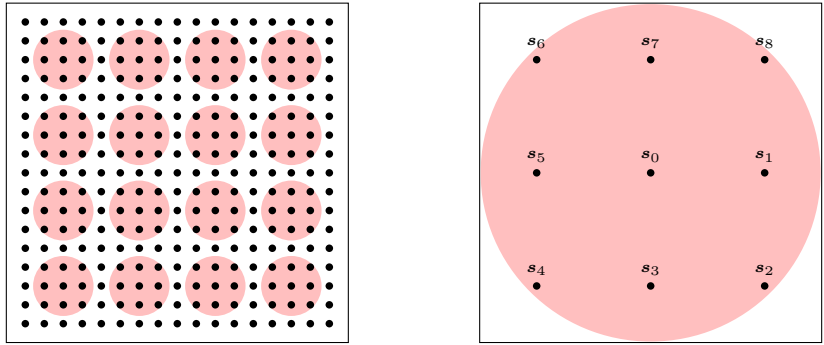
\begin{figure}[t]
\centering
\begin{subfigure}{0.32\textwidth}
	\centering
	\tiny	
	\vskip 0pt 				
	\begin{tikzpicture}[scale=0.25]
	\draw[draw = black] (-1,-1) rectangle (17,17);
		
	\foreach \i in {0,...,3}
	{
		\foreach \j in {0,...,3}
		{
			\node[circle, fill=pink, inner sep = 8] at (\i*4+2,\j*4+2){};
		}
	}
	\foreach \i in {0,...,16}
	{
		\foreach \j in {0,...,16}
		{
			\node[circle, fill, inner sep = 1] at (\i,\j){};      			
		}
	}		
	\end{tikzpicture}	
\end{subfigure} \hspace{25pt}
\begin{subfigure}{0.32\textwidth}
	\centering
	\tiny
	\vskip 0pt	
	\begin{tikzpicture}[scale=0.15]
	\draw[draw = black] (-15,-15) rectangle (15,15);			
	\node[fill = pink, circle, inner sep = 45] at (0,0){};		
	
	\node[label = $\psp_0$, circle, fill, inner sep = 1] (origin) at (0,0){};
	
	\node[label = $\psp_1$, circle, fill, inner sep = 1] (origin) at (10,0){};
	\node[label = $\psp_2$, circle, fill, inner sep = 1] (origin) at (10,-10){};
	\node[label = $\psp_3$, circle, fill, inner sep = 1] (origin) at (0,-10){};	
	\node[label = $\psp_4$, circle, fill, inner sep = 1] (origin) at (-10,-10){};	
	\node[label = $\psp_5$, circle, fill, inner sep = 1] (origin) at (-10,0){};	
	\node[label = $\psp_6$, circle, fill, inner sep = 1] (origin) at (-10,10){};	
	\node[label = $\psp_7$, circle, fill, inner sep = 1] (origin) at (0,10){};	
	\node[label = $\psp_8$, circle, fill, inner sep = 1] (origin) at (10,10){};					
	
		
	\end{tikzpicture}	
\end{subfigure}
	\caption{Evenly spaced local neighborhoods are chosen for the local estimates. Data from the nine locations in each neighborhood is used for the estimates. }
	\label{fig:localEstimation}
\end{figure}

An important feature of first estimating the parameters locally is that it can be used to construct the triangular mesh and choose the smoothness of the random field. 
It is important that the triangular mesh is fine enough to resolve the Gaussian random field approximation for a given $J[\warp^{-1}]$. During the global optimization, $J[\warp^{-1}]$ will vary and we should ideally update the mesh in every iteration of the numerical optimization. Unfortunately, meshing is one of the computationally most costly operations and it is therefore not feasible to remesh in each iteration. We instead precompute the mesh prior to the optimization, and use the local estimates to select a suitable maximum mesh resolution. 
Likewise, the smoothness parameter changes the structure of the FEM matrices drastically, and we therefore fix $\alpha$ based on the local estimates prior to estimating the other parameters globally. 

\section{Existence and estimation of $\dspace$-space}
\label{sec:dspace}
\label{sec:existenceDspace}

As mentioned in Section~\ref{sec:estimation}, we do not parametrize the model using the deformation function, $\warp^{-1}$, directly. This is not needed since the model only depends on $\warp^{-1}$ through the $\tilde{H}$ matrix. However, knowing $\warp^{-1}$ or $\dspace$ can in certain applications be of interest in itself since it might convey some physical interpretations, can be used to visualize deformations, and makes it possible to analyze data using techniques that relies on stationarity and/or isotropy. 
In order to acquire a $\dspace$-space from the estimated $\tilde{H}(\psp)$, three problems have to be taken care of.
Firstly, $J[\warp^{-1}]$ is not uniquely defined by $\tilde{H}$. 
A simple example of this is the stationary case where it is not possible to identify maps, $\warp^{-1}: \mathcal{G} \to \mathcal{D}$, that only differ by one being the mirror of the other with respect to an arbitrary plane through the origin. In this case, we can restrict $\warp^{-1}$ to the subset of all differentiable maps which have positive definite and symmetric Jacobian matrices. With that restriction we are able to identify $J[\warp^{-1}]$ from $\tilde{H}$ since they share eigenvectors and the eigenvalues of $J[\warp^{-1}]$ are just the inverse square root of those for $\tilde{H}$.

Secondly, for non-stationary fields where $J[\warp^{-1}]$ is spatially varying, the rows of $J[\warp^{-1}]$ correspond to vector fields. Since these vector fields are potential fields of the elements of $\warp^{-1}$ they have to be conservative. It is not clear to the authors what further restrictions this puts on the spatially varying, symmetric, and positive definite matrices $\tilde{H}$.

Thirdly, even given a matrix-valued function $J[\warp^{-1}]$ that is the Jacobian of some function $\warp^{-1}$, $\warp^{-1}$ is not necessarily injective.
For the sake of the argument, let us consider $\warp^{-1}$ as a mapping $\warp^{-1}:\gspace \to \dspace$, not necessarily injective. 
The inverse function theorem \citep{lit:spivak} states that there exists a differentiable injective function in an open set around a point, $\psp$, if $\determinant{J[\warp^{-1}](\psp)} \neq 0$. This is a necessary but unfortunately not sufficient condition to guarantee that $\warp^{-1}$ is injective. When $\warp^{-1}$ is not injective we say that $\dspace$ folds, i.e., several points in $\gspace$ maps to the same point in $\dspace$. 
Folding was also a problem using the thin-plate spline parameterization of $\warp^{-1}$ in the original paper of the deformation method~\citep{lit:sampson}. They reported issues with folding and recommended ameliorating this problem by forcing $\warp^{-1}$ to be smooth enough in their parameterization. For our model folding is not an issue since $\warp^{-1}$ only acts in a local sense through the value of the Jacobian. Hence, the inverse function theorem is enough for solutions to the SPDE to behave as intended. 
For us, folding is only a problem when the interest lies in properties of $\dspace$ itself. 
It is for us an open question how to find a practically useful parameterization of $\warp^{-1}$ in a way that ensures no folding. 

However, even if the estimated model folds, it should not occur in many places since $\tilde{H}$ is parameterized to vary smoothly. It is therefore possible to consider dividing $\gspace$ into subregions, $\{\gspace_k\}_k$. Each $\gspace_k$ are in turn mapped to a space $\dspace_k$ for which the data behaves as isotropic and stationary. 
Such a division can be used for instance to visualize the behavior of the estimated random field by merging overlapping subsets $\{\dspace_k\}_k$. This was done in Figure~\ref{fig:mapping} in order to visualize the deformation. 
    
Assuming that the acquired Jacobian matrix is a Jacobian of some function, $\warp^{-1}$, for a two dimensional model, the mapping of points between $\gspace$ and $\dspace$ can be computed as 
\begin{align}
F^{-1}(\psp) = F^{-1}(\psp_0) + 
\begin{bmatrix}
\int_0^1 a(\bs{r}(t)) dr_x(t)  \\
\int_0^1 c(\bs{r}(t)) dr_x(t) 
\end{bmatrix} +
\begin{bmatrix}
\int_0^1 c(\bs{r}(t)) dr_y(t) \\
\int_0^1 b(\bs{r}(t)) dr_y(t) 
\end{bmatrix}, 
\label{eq:getF}
\end{align}
where 
\begin{align}
J[\warp^{-1}](\psp) = 
\begin{bmatrix}
a(\psp) & c(\psp) \\
c(\psp) & b(\psp)
\end{bmatrix},
\end{align}
and the functions $a, b,$ and $c$ are defined through the spatially varying parameters $h_i$ of the $\tilde{H}$ matrix. 
By identifying one node in the FEM mesh on $\gspace$ as mapping to a distinct point in $\dspace$, the mapping of all nodes to $\dspace$ can be acquired by successively mapping each neighboring node in a recursive manner.
Points inside a triangle can be mapped to $\dspace$ using a linear approximation based on the location of the nodes of the triangle. 
We compute the integrals numerically by evaluating $a, b, c$, on the edges in the mesh and estimating the integral using Simpson's rule~\citep{lit:heath}.


\begin{figure}
\begin{center}
\includegraphics[width = 0.4\textwidth]{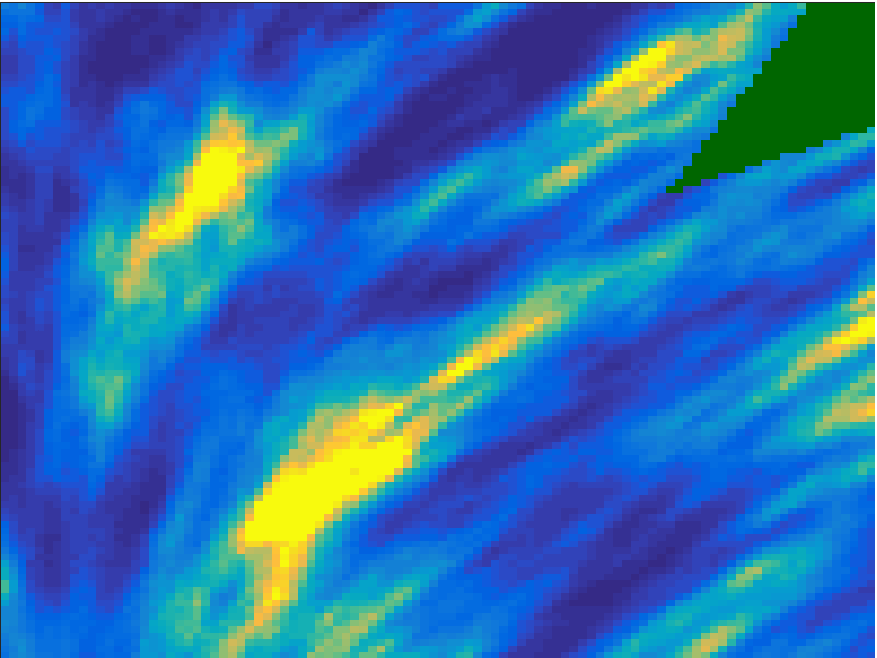}
\includegraphics[width = 0.4\textwidth,height=3.5cm]{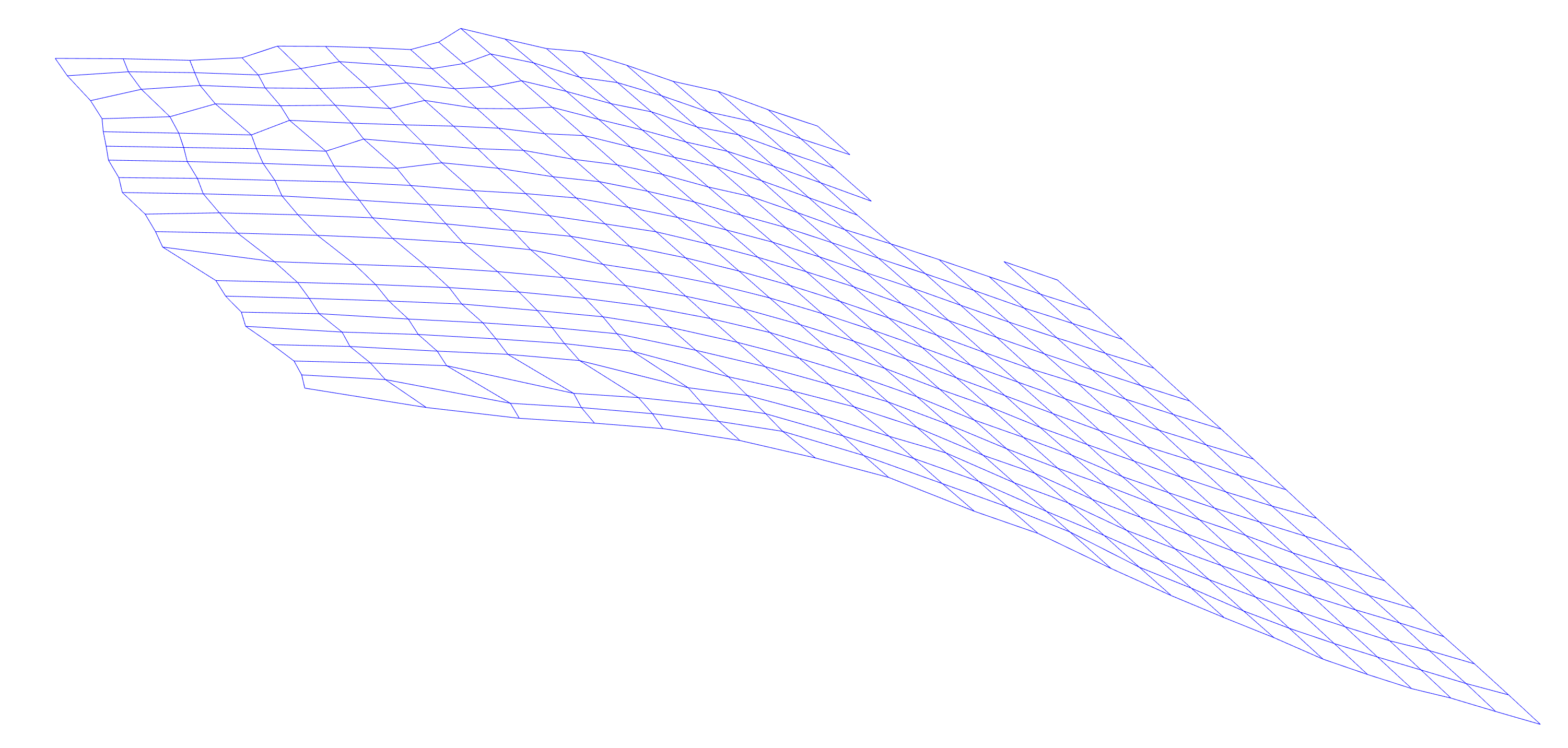}
\end{center}	
	\caption{Realization of a non-stationary and anisotropic Gaussian random field generated through the model of Section \ref{sec:model} (left), and $\dspace$ acquired from estimating the parameters of the model from data (right). The locations of a regular grid in $\gspace$ is shown in $\dspace$. }
	\label{fig:example_nonstationary}
\end{figure}

An example of an acquired $\dspace$ is shown in Figure~\ref{fig:example_nonstationary}, where a realization of a non-stationary and anisotropic Gaussian random field on a non-rectangular domain, defined through the model of Section~\ref{sec:model}, is shown. Based on $182$ realizations of the field, the parameters of the model was estimated using the procedure described in Section~\ref{sec:estimation}. The $\dspace$ region was estimated using the algorithm described above, and is shown in the right part of the figure. When observing such a figure, interpretation of the deformation comes from our knowledge that the connected nodes was on a regular grid in $\gspace$ before the deformation.




\section{Proofs}
\label{sec:proofs}

\subsection{A SPDE equivalent to a deformed Mat\'{e}rn field}
To prove Theorem~\ref{thm:spde}, we first need some additional notation. As previously mentioned, $\warp : \mathcal{D} \mapsto \mathcal{G} := F(\mathcal{D}) \subset \mathbb{R}^d$ is a bijective and differentiable function.
We consider the slightly more general operator $\hat{\mathcal{L}}_c = (c^2 - \hat{\nabla}\cdot \hat{\nabla})$ on $\mathcal{D}$ for a constant $c>0$. Similarly, we define an operator $\mathcal{L}_c = c^2 - \kappa^{-2} \nabla \cdot H \nabla$ on $\mathcal{G}$.

Let $\{\lambda_j\}_{j\in\mathbb{N}}$ denote the eigenvalues of $\hat{\mathcal{L}}_c$ in nondecreasing order, and let $\{e_j\}_{j\in\mathbb{N}}$ denote the corresponding eigenfunctions. As in \citet{bolinIMA}, we use the spectral definition of the operator $\hat{\mathcal{L}}_c^{\beta} : \mathscr{D}(\hat{\mathcal{L}}^{\beta}) \rightarrow L_2(\mathcal{D})$.  That is, for $\beta>0$ and $\phi\in \mathscr{D}(\hat{\mathcal{L}}_c^{\beta})$, the action of the operator is defined by
$$
\hat{\mathcal{L}}_c^{\beta}\phi = \sum_{j\in\mathbb{N}} \lambda_j^{\beta} \langle \phi, e_j \rangle e_j.
$$
We define the space $\dot{H}^{2\beta} := \mathscr{D}(\hat{\mathcal{L}}_c^{\beta})$, which is a Hilbert space with inner product $\langle \hat{\mathcal{L}}_c^{\beta}\phi, \hat{\mathcal{L}}_c^{\beta}\psi \rangle$. Further, $\dot{H}^{-2\beta}$ is defined as the dual space of $\dot{H}^{2\beta}$ and $\dot{H}^{0} := L_2(\mathcal{D})$ \citep[see][for further details]{bolinIMA}. 

Lemma 2.1 in \citet{bolinIMA} shows that for $s\in\mathbb{R}$, $\hat{\mathcal{L}}_c$ can be extended to an isometric isomorphism $\hat{\mathcal{L}}_c : \dot{H}^s \rightarrow \dot{H}^{s-2\beta}$. This means that if $\hat{g} \in \dot{H}^s$ $\mathbb{P}$-a.s., for $s\geq - 2\beta$, then $\hat{\mathcal{L}}_c^{\beta} \hat{X} = \hat{g}$ has a unique solution in $L_2(\Omega, L_2(\mathcal{D}))$. In particular, we have by Proposition 2.3 in \citet{bolinIMA} that white noise on $\mathcal{D}$ is in $\dot{H}^{-d/2-\epsilon}$ $\mathbb{P}$-a.s. for any $\epsilon>0$. For $\hat{g}$, we define the corresponding function $g := \hat{g} \circ F^{-1}$.

The proof of Theorem~\ref{thm:spde} will make use of three lemmas. 
The following lemma connects the two operators $\hat{\mathcal{L}}_c$ and $\mathcal{L}_c$.

\begin{lemma}\label{lemma:alpha2}
Assume that $\hat{g} \in \dot{H}^{s}$ $\mathbb{P}$-a.s. for $s\geq -2$ and let $\hat{X}$ be the solution to 
\begin{align}
\hat{\mathcal{L}}_c \hat{X} = \hat{g} \quad \mbox{in $\mathcal{D}$}.
\end{align}
Then, $X(\psp) := \hat{X}(F^{-1}\psp)$ is the solution to
\begin{align}
\mathcal{L}_c X = g \quad \mbox{in $\mathcal{G}$}.
\end{align}
\end{lemma}
\begin{proof}
The operator $\hat{\mathcal{L}}_c$ induces a symmetric, continuous and coercive bilinear form $B_{\hat{\mathcal{L}}_c}$ on $V_\mathcal{D} := H_0^1(\mathcal{D})$,
$$
B_{\hat{\mathcal{L}}_c} : V_\mathcal{D} \times V_\mathcal{D} \rightarrow \mathbb{R}, \qquad B_{\hat{\mathcal{L}}_c}(\hat{u},\hat{v}) := c^2 \langle \hat{u},\hat{v} \rangle_{L_2(\mathcal{D})} +  \langle \hat{\nabla} \hat{u},\hat{\nabla} \hat{v} \rangle_{L_2(\mathcal{D})}.
$$
The solution to $\hat{\mathcal{L}}_c \hat{u} = \hat{g}$ in $\mathcal{D}$ can be viewed as the unique element $\hat{u} \in V_\mathcal{D}$ satisfying 
$B_{\hat{\mathcal{L}}_C}(\hat{u},\hat{v}) = \hat{g}(\hat{v})$ for all $\hat{v} \in V_\mathcal{D}$. 

To obtain the desired result we perform a change of variables in this variational formulation.  First, we have
\begin{align*}
\langle \hat{u},\hat{v} \rangle_{L_2(\mathcal{D})} &= \int_{\mathcal{D}}\hat{u}(\hat{\psp})\hat{v}(\hat{\psp}) d\hat{\psp} = \int_{\mathcal{G}} \frac{ \hat{u}(F^{-1}\psp) \hat{v}(F^{-1}\psp)}{|J[F](F^{-1}\psp)|}d\psp\\
&= \int_{\mathcal{G}} u(\psp)v(\psp)\kappa^2(\psp)d\psp
= \langle \kappa^{2}u,v \rangle_{L_2(\mathcal{G})},
\end{align*}
where we used that $J[F](F^{-1}(\psp)) = J[F^{-1}]^{-1}(\psp)$ due to the inverse function theorem. We also defined $v(\psp) = \hat{v}(F^{-1}(\psp))$, and one should note that every function $v\in V_{\mathcal{G}} := H_0^1(\mathcal{G})$ can be represented in this way. Now, if we let $\nabla$ denote the gradient on $\mathcal{G}$,
\begin{align*}
\langle \hat{\nabla} \hat{u},\hat{\nabla} \hat{v} \rangle_{L_2(\mathcal{D})} &= \int_{\mathcal{D}} \left[\hat{\nabla} \hat{u}(\hat{\psp})\right]^T \hat{\nabla} \hat{v}(\hat{\psp}) d\hat{\psp}\\
&= \int_{\mathcal{G}} \frac{ 
\left[ J[F]^T(F^{-1}\psp)\nabla \hat{u}(F^{-1}\psp) \right]^T 
\left[ J[F]^T(F^{-1}\psp)\nabla \hat{v}(F^{-1}\psp) \right]
}{ \determinant{J[F](F^{-1}\psp)}} d\psp 
\\
&= \int_{\mathcal{G}} 
\left[ J[F^{-1}]^{-T}(\psp) \nabla u(\psp) \right]^T  
\left[ J[F^{-1}]^{-T}(\psp) \nabla v(\psp) \right]  
\kappa^2(\psp)d\psp 
\\
&= \left\langle \kappa J[F^{-1}]^{-T}\nabla u,\kappa J[F^{-1}]^{-T}\nabla v \right\rangle_{L_2(\mathcal{G})}.
\end{align*}
Thus, the bilinear form, $B_{\hat{\mathcal{L}}_c}$, corresponds to a bilinear form, $B_{\mathcal{L}} : V_\mathcal{G} \times V_\mathcal{G} \rightarrow \mathbb{R}$ 
$$
B_{\mathcal{L}}(u,v) := c^2 \langle \kappa^{2}u,v \rangle_{L_2(\mathcal{G})} +  \left\langle \kappa J[F^{-1}]^{-T}\nabla u,\kappa J[F^{-1}]^{-T}\nabla v \right\rangle_{L_2(\mathcal{G})},
$$
induced by the operator 
$\kappa^{2} \mathcal{L}_c =  \kappa^{2} c^2 -  \nabla \cdot H \nabla$.

Performing the same change of variables on the right-hand side yields
\begin{align*}
\hat{g}(\hat{v}) &= \int_{\mathcal{D}} \hat{v}(\hat{\psp}) \hat{g}(\hat{\psp}) d\hat{\psp} = \int_{\mathcal{G}} \hat{v}(F^{-1}\psp) \hat{g}(F^{-1}\psp)\kappa^2(\psp)d\psp \\
&= \int_{\mathcal{G}} v(\psp) \hat{g}(F^{-1}\psp)\kappa^2(\psp)d\psp.
\end{align*}
This implies $\mathcal{L}_c X = g$ in $\mathcal{G}$ since $\kappa(\psp) > 0, \forall \psp \in \mathcal{G}$.
\end{proof}

In Equation~\eqref{eq:Dspde} the right-hand side is white noise, i.e., $\hat{g} = \hat{\mathcal{W}}$. The following lemma identifies the distribution of $g$ for this case. 
\begin{lemma}\label{lemma:whitenoise}
If $\hat{g}$ is a white noise on $\mathcal{D}$, then $g \overset{D}{=} \kappa^{-1} \mathcal{W}$. That is, $g$ is a white noise on $\mathcal{G}$ scaled by $\kappa^{-1}$.
\end{lemma}

\begin{proof}
Let $\hat{\mathcal{W}}$ be a white noise on $\mathcal{D}$.
Due to the properties of white noise and a change of variables,
\begin{align}
g(v) &:= (\hat{g} \circ F^{-1}) (v) = \int_{\mathcal{G}} \hat{\mathcal{W}}(F^{-1}\psp) v(\psp) d\psp
= \int_{\mathcal{D}} \hat{\mathcal{W}}(\hat{\psp}) v(F \hat{\psp}) \kappa^{-2}(F \hat{\psp}) d\hat{\psp}
\\
&= \hat{\mathcal{W}}\left( \hat{v} (\kappa^{-2} \circ F) \right)   \sim \N \left( 0, \int_{\mathcal{D}} \left( \hat{v}^2(\hat{\psp}) \kappa^{-2}(F\hat{\psp}) \right)^2 d\hat{\psp} \right) 
\\
&= \N \left( 0, \int_{\mathcal{G}} \left( v(\psp) \kappa^{-2}(\psp) \right)^2 \kappa^2(\psp) d\hat{\psp} \right) 
= \N \left( 0, \int_{\mathcal{G}} ( v(\psp) \kappa^{-1}(\psp) )^2 d\hat{\psp} \right).
\end{align}
\end{proof}

One can note that Lemma~\ref{lemma:alpha2} together with Lemma~\ref{lemma:whitenoise} can be used to show the result of Theorem~\ref{thm:spde} in the non-fractional case. It is however not clear if $\hat{\mathcal{L}}_c^{\beta}$ maps to $\mathcal{L}_c^{\beta}$ for arbitrary $\beta$. The following lemma is necessary in order to show this. 

\begin{lemma}\label{lemma:fracpow}
Let $\beta := \alpha/2 \in (0, 1)$ and assume that $\hat{g} \in \dot{H}^{s}$, $\mathbb{P}$-a.s., for $s\geq -2\beta$. Let 
$\hat{X}$ be the solution to
\begin{align}
    \hat{\mathcal{L}}_c^{\beta} \hat{X} = \hat{g} \quad \mbox{in $\mathcal{D}$}.
\end{align}
Then, $X(\psp) := \hat{X}(F^{-1}\psp)$ is the solution to
\begin{align}
    \mathcal{L}_c^{\beta} X = g \quad \mbox{in $\mathcal{G}$}.
\end{align}
\end{lemma}

\begin{proof}
Introducing the composition operator $\mathcal{C}_F: V_{\mathcal{D}} \rightarrow V_{\mathcal{G}}$ defined by $\mathcal{C}_Fv = v\circ F$, we can state the non-fractional result of Lemma \ref{lemma:alpha2} as 
\begin{equation}\label{eq:comp}
 X = \mathcal{L}_c^{-1} g = \mathcal{C}_F \hat{\mathcal{L}}_c^{-1} \hat{g}.
\end{equation}
Now, in the fractional case we have defined the fractional power of the operator using the spectral definition. We can therefore use the following representation of the fractional inverse from the Dunford-Taylor calculus by \cite{lit:balakrishnan}. For $\beta \in (0,1)$,
$$
\hat{X} = (c^2 - \hat{\nabla} \cdot \hat{\nabla})^{-\beta} \hat{g} =  \frac{2\sin(\pi\beta)}{\pi}\int_0^{\infty} t^{2\beta-1} \left(I + t^2( c^2 - \hat{\nabla} \cdot \hat{\nabla} )\right)^{-1}dt \hat{g}.
$$
What we need to show is that the corresponding fractional model on $\mathcal{G}$,
$$
X = \left( c^2 - \kappa^{-2} \nabla \cdot H \nabla \right)^{-\beta} g
$$
is equal to $\mathcal{C}_F (c^2 - \hat{\nabla} \cdot \hat{\nabla})^{-\beta} \hat{g}$. Using the Balakrishnan  representation of the fractional inverse again, we have
\begin{align*}
X &= \left( c^2 - \kappa^{-2} \nabla \cdot H \nabla \right)^{-\beta} g \\
&= \frac{2\sin(\pi\beta)}{\pi}\int_0^{\infty} t^{2\beta-1} \left(I + t^2( c^2 - \kappa^{-2} \nabla \cdot H \nabla )\right)^{-1}dt g \\
&= \frac{2\sin(\pi\beta)}{\pi}\int_0^{\infty} t^{2\beta-3} \left(t^{-2} + c^2 - \kappa^{-2} \nabla \cdot H \nabla\right)^{-1} g dt.
\end{align*}
At this stage, we can use Lemma~\ref{lemma:alpha2} with operator $\mathcal{L}_C$ and $C^2 := t^{-2} + c^2$ to change from $\mathcal{G}$ to $\mathcal{D}$ in the integral. Also using that $\mathcal{C}_F$ is a bounded linear operator, we get 
\begin{align*}
X
&= \frac{2\sin(\pi\beta)}{\pi}\int_0^{\infty} t^{2\beta-3} \mathcal{C}_F \left(C^2 - \hat{\nabla} \cdot \hat{\nabla}\right)^{-1} \hat{g} dt \\
&= \mathcal{C}_F \frac{2\sin(\pi\beta)}{\pi}\int_0^{\infty} t^{2\beta-1} (I + t^2\left( c^2 -  \hat{\nabla} \cdot \hat{\nabla} \right))^{-1} \hat{g} dt\\
&= \mathcal{C}_F (c^2-\nabla\cdot\nabla)^{-\beta} \hat{g}.
\end{align*}
\end{proof}

Given these lemmas we are now ready for the proof of Theorem~\ref{thm:spde}.

\begin{proof}[Proof of Theorem~\ref{thm:spde}]
Recall that for $s\in\mathbb{R}$, $\hat{\mathcal{L}}_c$ is an isometric isomorphism $\hat{\mathcal{L}}_c : \dot{H}^s \rightarrow \dot{H}^{s-2\beta}$. Further recall that white noise on $\mathcal{D}$ is in $\dot{H}^{-d/2-\epsilon}$ $\mathbb{P}$-a.s. for any $\epsilon>0$. With these facts, Lemma~\ref{lemma:alpha2} together with Lemma~\ref{lemma:whitenoise} directly give the result for $\alpha = 2$. Also, Lemma~\ref{lemma:fracpow} give the result for all $\alpha \in (d/2, 1]$.

The spectral definition of fractional operators yields that the operators adhere to the law
\begin{align}
    \mathcal{L}^{\alpha}X = \mathcal{L}^{\alpha-1} \left(\mathcal{L} X \right).
\end{align}
Hence, the result for a general $\alpha > d/2$ readily follows by induction and the fact that $\mathcal{L}_c^{\beta} X = \noise$ can be represented as 
$\mathcal{L}_c^{\tilde{\beta}} u = X_0$,
where $\tilde{\beta}$ is the integer part of $\beta$ and 
$\mathcal{L}^{\beta - \tilde{\beta}} X_0 = \noise,$
with $\beta - \tilde{\beta} \in (0,1)$. 
\end{proof}

\subsection{Exceedance probability}

\begin{proof}[Proof of Proposition \ref{thm:exceedanceProb}]
Let $N_T(v)$ be the number of upcrossings for $X_\gamma(t)$ at a threshold value, $u$, on the route. Then by Rice's method of moments \citep{lit:piterbarg} we have 
\begin{align}
\mathbb{P}\left[\max_{t\in[0,T]} X_\gamma(t)>u\right] 
 &= \mathbb{P}\left[X_\gamma(0)>u\right] + \mathbb{P}\left[ (N_T(u) > 0) \cap ( X_\gamma(0)\le u ) \right]  \\
&\le \prob{X_\gamma(0)>u} + \expect{N_T(u)}.
\end{align}
Rice's formula \citep{lit:azais, lit:rice, lit:rice2} give an explicit expression,
\begin{align}
\expect{N_T(u)}=\int_0^T \condexpect{|\dot{\rv}_\gamma(t)| \indicator{\dot{\rv}_\gamma(t) \ge 0} }{ \rv_\gamma(t)=u}f_{\rv_\gamma(t)}(u)dt,
\label{eq:ricesFomula}
\end{align}
where $\dot{\rv}_\gamma$ denotes the mean square derivative of $\rv_\gamma(t)$ and $\mathbb{I}$ is the indicator function.  
Now, if $\rv_\gamma$ would have constant variance, the derivative of the process at $t$ would not be correlated with the value of the process at $t$~\citep[section 5.6]{lit:adler}. To make use of this property, we introduce the standardardized Gaussian process $W(t) := \frac{\rv_\gamma(t) - \mu_\gamma(t)}{\std_\gamma(t)}$ and get 
\begin{align}
\dot{\rv}_\gamma(t) = \dot{W}(t)\std_\gamma(t) + W(t)\dot{\std}_\gamma(t) + \dot{\mu}_\gamma(t).
\end{align}
$W(t)$ and $\dot{W}(t)$ for a fixed $t$ are independent. Hence, 
\begin{align}
    \condexpect{|\dot{\rv}_\gamma(t)| \indicator{\dot{\rv}_\gamma(t) \ge 0} }{ \rv_\gamma(t)=u}
    &= \std_\gamma(t) \expect{ \dot{W}(t) \indicator{\dot{W}(t) \ge -a(t)} } + \std_{\gamma}(t) a(t),
\end{align}
where $a =  \frac{\dot{\mu}_\gamma(t)}{\std_{\gamma}(t)} + \dot{\std}_{\gamma}(t)\frac{u-\mu_{\gamma}(t)}{\std_{\gamma}^2(t)}$.
Since $\dot{W}(t)$ is a centered random process,
\begin{align}
    \expect{ \dot{W}(t) \indicator{\dot{W}(t) \ge -a(t)} }
    &= \std_{\dot{W}}(t) \phi\left( \frac{a(t)}{\std_{\dot{W}}(t)} \right).
\end{align}
Now, using Equation \eqref{eq:ricesFomula}, the fact that $f_{X_{\gamma}(t)}(u) = \frac{1}{\std_{\gamma}(t)}\phi\left( \frac{u-\mu_{\gamma}(t)}{\std_{\gamma}(t)} \right)$, and plugging the result into the first inequality finishes the proof.
\end{proof}

The general choice of model for $\mu_{\gamma}(t)$ and $\std_{\gamma}(t)$, as well as the choice of numerical approximation of the integral in Proposition~\ref{thm:exceedanceProb}, are beyond the scope of this work. 
However, a quantity present in Proposition~\ref{thm:exceedanceProb} which is highly relevant to the work of this paper is the variance of the mean square derivative of $W(t)$, i.e., $\std^2_{\dot{W}}(t)$.
We will now show how this variance can be computed when $X(\psp,t)$ is a spatio-temporal Gaussian random process defined using a stationary covariance function and a (local) spatial deformation. 
By local we mean that there exists a function, $F^{-1}$, satisfying the inverse function theorem~\citep{lit:spivak} over any Lebesgue measurable set of the spatial domain. This means that the proof will hold for any deformation model~\citep{lit:sampson}, any locally deformed Mat\'{e}rn SPDE model, but also any other local deformation model. 

In Section~\ref{sec:model} we referred to $\dspace$ as the space where our spatial model would be stationary, isotropic, and for which the correlation was Mat\'{e}rn. 
We now make a similar definition for a spatio-temporal process. Assume that $X$ is defined through a local deformation approach using a function $F^{-1}(\psp, t) : \gspace_s \times \gspace_t \rightarrow \dspace_s \times \dspace_t (\psp, t)$ such that the process on $\dspace_s \times \dspace_t (\psp, t)$ is stationary but not necessarily isotropic. Notice how the relaxation of a global deformation require the product space $\dspace_s \times \dspace_t (\psp, t)$ to be indexed by points in $\gspace_s \times \gspace_t$, i.e. a local mapping and a local space for each neighbourhood.
Let $\tilde{X}'_x, \tilde{X}'_y, ..., \tilde{X}'_t$ denote the partial derivatives of $\tilde{X}$. Then, $\std_{\dot{W}}(t)$ can be computed using the following proposition.
\begin{prop}\label{thm2}
Let $\sigma^2_{ij} = \Cov(\tilde{X}'_i(\tilde{\psp},\tilde{t}), \tilde{X}'_j(\tilde{\psp},\tilde{t}))$, for $i,j \in \{1, ..., d, d+1\}$, i.e., the set of all spatial dimensions and the temporal. With $\Sigma$ the corresponding matrix $\Sigma_{ij} = \{\std^2_{ij}\}$ we have
\begin{align}
\std^2_{\dot{W}}(t) &= \dot{\gamma}(t)^T J[\warp^{-1}](\gamma(t))  \Sigma J[\warp^{-1}]^{T}(\gamma(t)) \dot{\gamma}(t),
\end{align}
where $\gamma(t) = [x(t), y(t), ..., t]^T$ and $\dot{\gamma}(t) = [ \dot{x}(t), \dot{y}(t), ..., 1 ]^T$ is the velocity of the curve at time $t$.
\end{prop}

\begin{proof}
By the chain rule we get
\begin{align}
\dot{W}(t) = \dot{\gamma}(t) \cdot \left[ \nabla X(\gamma(t)), \dot{X}(\gamma(t)) \right]^T = \dot{\gamma}(t)^T J\left[\warp^{-1}\right]^{T}(\gamma(t)) \left[ \tilde{\nabla} \tilde{X}(\tilde{\gamma}(t)), \dot{\tilde{X}}(\tilde{\gamma}(t)) \right]^T.
\end{align}
Due to stationarity of $\tilde{X}(\tilde{\gamma}(t))$, the marginal distribution with respect to space-time of the partial derivatives of $\tilde{X}(\tilde{\gamma}(t))$ is a $d+1$-dimensional Gaussian distribution independent of $\gamma(t)$. 
Let $\Sigma$ denote the covariance matrix of this multivariate Gaussian distribution. Then, 
\begin{align}
    \Sigma := \expect{ \left[ \tilde{\nabla} \tilde{X}(\tilde{\gamma}(t)), \dot{\tilde{X}}(\tilde{\gamma}(t)) \right]^T \left[ \tilde{\nabla} \tilde{X}(\tilde{\gamma}(t)), \dot{\tilde{X}}(\tilde{\gamma}(t)) \right] } = \Sigma.
\end{align}

The proof is concluded by the fact that 
\begin{align}
\std^2_{\dot{W}}(t) &= \expect{ \dot{W}(t) \dot{W}(t)^T } = \dot{\gamma}(t)^T J[\warp^{-1}]^{T}(\gamma(t)) \Sigma J[\warp^{-1}](\gamma(t)) \dot{\gamma}(t).
\end{align}
\end{proof}

The point of Proposition~\ref{thm2} is that when a model i specified as having a stationary covariance function given a (local) deformation, that covariance function is usually available in explicit form and well understood. This means that, in most cases, $\Sigma$ can be computed explicitly and hence also $\std_{\dot{W}}(t)$.
It should be noted that all pairs of sub-dimensions that are isotropic in $\dspace$, given the other dimensions being fixed, have partial derivatives of $\tilde{X}$ that are independent of each other~\citep[section 5.7]{lit:adler}. 
For a covariance function in the deformed space that is not only stationary but isotropic as well, this yields
\begin{align}
\std^2_{\dot{W}}(t) &= \sigma^2_{ii} \dot{\psp}(t)^T J[\warp^{-1}]^{T}(\psp(t)) J[\warp^{-1}](\psp(t)) \dot{\psp}(t),
\end{align}
for any $i \in \{1, ..., d+1\}$ since $\std_{xx}^2 = \std_{yy}^2 = ... = \std_{tt}^2$ due to isotropy.
Moreover, for the case of a spatio-temporal Mat\'{e}rn covariance in $\dspace$, 
$$\std^2_{ii} = \frac{ \std^2 \kappa^2 }{ 2 (\nu-1) }, \quad \forall i\in \{1, ..., d+1\}.$$ 
This can be derived from the spectral density of the Mat\'{e}rn covariance function~\citep{lit:stein} together with the fact that the second spectral moment in the direction of dimension $i$ equals the variance of the partial derivative process in the same direction~\citep[section 5.5]{lit:adler}.
Notice that we need $\nu>1$ for the result to hold since the field is otherwise not mean-square differentiable.
\begin{proof}[Proof of Corollary~\ref{cor:stdDeriv}]
For the special case of the model of Section~\ref{sec:model} that is constant over time, the elements in the last row and column of $\Sigma$ are zero since $\dot{\tilde{X}}(\tilde{\psp}, \tilde{t}) \equiv 0$. In this model, the spatial dimensions have an isotropic Mat\'{e}rn covariance structure with unit dampening and unit variance, hence
\begin{align}
\std^2_{\dot{W}}(t) &= \frac{\damp^2(t)}{2(\nu-1)} \dot{\psp}(t)^T H^{-1}(\psp(t)) \dot{\psp}(t).
\end{align}
It should be clarified that $\kappa(t)$ above refer to the function $\kappa(\psp(t))$ from Section~\ref{sec:model} over the curve, $\gamma$. It is not the dampening in the deformed space---which we know is constant and equal to one.
\end{proof}

\bibliography{shell}

\end{document}